\documentclass[11pt,a4paper]{article}

\usepackage{amsmath,amsfonts,amssymb,amsthm}
\usepackage[utf8]{inputenc}
\usepackage[hang,center,nooneline]{subfigure}
\usepackage{a4wide,natbib,color,graphicx,enumerate,url,hyperref,cancel,paralist,dsfont,epstopdf,accents}
\usepackage[titletoc, title]{appendix}
\usepackage{subfigure}

\usepackage{authblk}

\usepackage{caption}

\definecolor{mblue}{rgb}{0,0.4470,0.7410}
\definecolor{morange}{rgb}{0.8500,0.3250,0.0980}
\definecolor{myellow}{rgb}{0.9290,0.6940,0.1250}
\definecolor{mgreen}{rgb}{0.25,0.5,0.25}

\hypersetup{pdfborder={0 0 0},colorlinks=true,linkcolor=mblue,urlcolor=mblue,citecolor=mblue}

\newtheorem{theorem}{Theorem}[section]

\newtheorem{proposition}[theorem]{Proposition} 
\newtheorem{lemma}[theorem]{Lemma}
\newtheorem{remark}[theorem]{Remark}
\theoremstyle{remark}

\numberwithin{equation}{section}

\newcommand{\eps}{\varepsilon}
\newcommand{\indic}{\mathds{1}}
\newcommand{\reals}{\mathbb{R}}
\newcommand{\ints}{\mathbb{Z}}
\newcommand{\s}{\mathbb{S}}
\newcommand{\ms}{\mathbb{M}}
\newcommand{\bd}{\mathbb{F}}
\newcommand{\prob}{\operatorname{P}}
\newcommand{\expec}{\operatorname{E}}

\newcommand{\tv}{\operatorname{TV}}

\renewcommand{\ge}{\geqslant}
\renewcommand{\le}{\leqslant}
\renewcommand{\geq}{\geqslant}
\renewcommand{\leq}{\leqslant}
\newcommand{\vcirc}{\accentset{\circ}{v}}
\newcommand{\dcirc}{\accentset{\circ}{\Delta}}

\newcommand{\KL}{\operatorname{KL}}

\newcommand{\diam}{\operatorname{diam}}
\newcommand{\KLball}{\operatorname{B_{KL}}}
\newcommand{\Lball}{\operatorname{B}_{1}}
\newcommand{\Sball}{\operatorname{B}_{\infty}}

\newcommand\blfootnote[1]{%
  \begingroup
  \renewcommand\thefootnote{}\footnote{#1}%
  \addtocounter{footnote}{-1}%
  \endgroup
}

\begin{document}

\title{Bayesian Nonparametrics for Directional Statistics}

\title{Bayesian Nonparametrics for Directional Statistics}
\author{Olivier Binette}
\author{Simon Guillotte}
\affil[]{Université du Québec à Montréal}

\maketitle

\blfootnote{Email: \texttt{binette.olivier@courrier.uqam.ca}}

\begin{abstract}
We introduce a density basis of the trigonometric polynomials that is suitable to mixture modelling. Statistical and geometric properties are derived, suggesting it as a circular analogue to the Bernstein polynomial densities. Nonparametric priors are constructed using this basis and a simulation study shows that the use of the resulting Bayes estimator may provide gains over comparable circular density estimators previously suggested in the literature.

From a theoretical point of view, we propose a general prior specification framework for density estimation on compact metric space using sieve priors. This is tailored to density bases such as the one considered herein and may also be used to exploit their particular shape-preserving properties. Furthermore, strong posterior consistency is shown to hold under notably weak regularity assumptions and adaptative convergence rates are obtained in terms of the approximation properties of positive linear operators generating our models. 
\end{abstract}

\section{Introduction}
\label{S:Intro}

There is increasing interest in the statistical analysis of non-euclidean data, such as data lying on a circle, on a sphere or on a more complex manifold or metric space. Applications range from the analysis of seasonal and angular measurements to the statistics of shapes and configurations \citep{RS01, BB12}. In bioinformatics, for instance, an important problem is that of using the chemical composition of a protein to predict the conformational angles of its backbone \citep{Allazikani2001}. Bayesian nonparametric methods, accounting for the wrapping of angular data, have been successfully applied in this context \citep{Lennox2009, Lennox2010}.

Directional statistics deals in particular with univariate angular data and provides basic building blocks for more complex models.
Among the most commonly used model for the probability density function of a circular random variable is the von Mises density defined by 
$$
    u \mapsto \exp(\kappa \cos(u - \mu)) / (2\pi I_0(\kappa)),
$$
where $\mu$ is the circular mean, $\kappa  > 0$ is a shape parameter and $I_0$ is the modified Bessel function of the first kind and order $0$. This function is nonnegative, $2\pi$-periodic and integrates to one on the interval $[0, 2\pi)$. It can be regarded a circular analogue to normal distribution \citep{RS01} (see also \cite{CLB12} for a comparison with the geodesic normal distribution). Mixtures of von Mises densities and other log-trigonometric densities are also frequently used \citep{Kent83}. Another natural approach is to model circular densities using trigonometric polynomials
\begin{equation}\label{eq:trig_poly}
    u \mapsto \frac{1}{2\pi} + \sum_{k=1}^{n} (a_k \cos(ku) + b_k \sin(ku)). 
\end{equation}
These densities have tractable normalizing constants, but the coefficients $a_k$ and $b_k$ must be constrained as to ensure nonnegativity \citep{Fejer1915, FD04}. 

For a review of common circular distributions, see \cite{MJ00, RS01}. Notable Bayesian approaches to directional statistics problems include \cite{Ghosh:2003, MM08, Ravindran2011, Hernandez2017}.

In this paper, we introduce a basis of the trigonometric polynomials \eqref{eq:trig_poly} consisting only of probability density functions. Properties shown in Section \ref{S:Model}, such as its shape-preserving properties, suggest it as a circular analogue to the Bernstein polynomial densities and we argue that it is particularly well suited to mixture modelling. In Section \ref{S:Prior}, we use this basis to devise nonparametric priors on the space of bounded circular densities. We compare their posterior mean estimates to other density estimation methods based on the usual trigonometric representation \eqref{eq:trig_poly} in Section \ref{sec:comparison}. 

An important aspect of nonparametric prior specification is the posterior consistency property, which entails almost sure convergence (in an appropriate topology) of the posterior mean estimate. In Section \ref{sec:strong_consistency}, we thus develop a general prior specification framework that immediately provides consistency of a class of sieve priors for density estimation on compact metric spaces. Particular instances of this framework appeared previously in the literature. For instance, \cite{Petrone2002} obtained consistency of the Bernstein-Dirichlet prior on the set of continuous densities on the interval $[0,1]$. More recently \cite{Xing2009} (see also \cite{walker2004, Lijoi05}) have obtained a simple condition for models of this kind ensuring consistency on the Kullback-Leibler support of the prior. As an application, they quickly revisit the problem of \cite{Petrone2002} but without discussing what contains the Kullback-Leibler support. Our main contribution here is the proof that the Kullback-Leibler support of the priors specified in our framework contains every bounded density. Furthermore, we show in Section \ref{sec:adaptative} how our framework may be used to obtain posterior contraction rates. The results are related to those of \cite{Ghosal2001, Kruijer2008} in the case of the Bernstein-Dirichlet prior but are stated with more generality. They express posterior contraction rates in terms of a balance between the dimension of the sieves and their approximation properties, as they are accounted for by a sequence of positive linear approximation operators.

\section{\textit{De la Vallée Poussin} mixtures for circular densities}
\label{S:Model}
\subsection{The basis}
We propose the basis $\mathcal{B}_n$ for $2\pi$-periodic densities of circular random variables given by
\begin{equation}
 \label{eq:basis}
C_{j,n}(u)=\frac{2^{2n}}{2\pi \binom{2n}{n}}\left(\frac{1+\cos\left(u-\frac{2\pi 
j}{2n+1}\right)}{2}\right)^n, \quad u \in \reals, \quad j=0,\ldots,2n,
\end{equation}


\begin{figure}[h]
  \centering
  \begin{minipage}[b]{0.45\textwidth}
    \includegraphics[width=\textwidth]{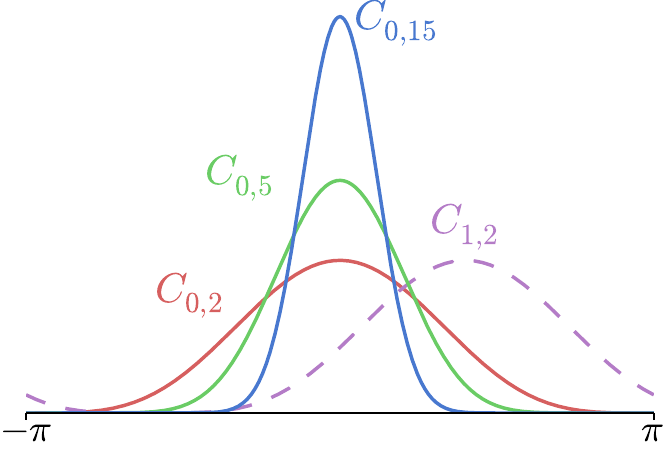}
  \end{minipage}
  \hfill
  \begin{minipage}[b]{0.45\textwidth}
    \includegraphics[width=\textwidth]{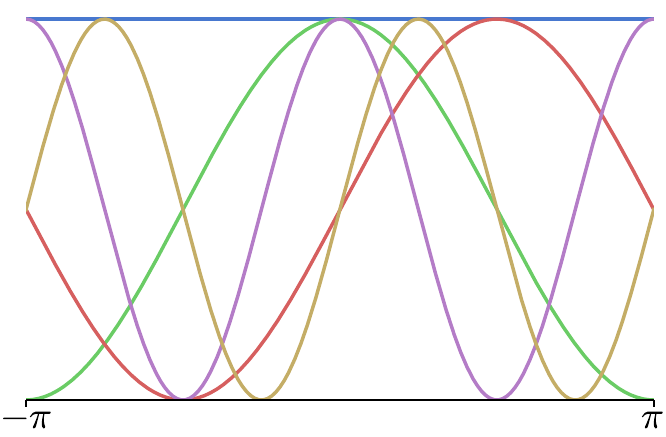}
  \end{minipage}
  \caption{Comparison between De la Vallée Poussin basis densities (left) and the usual trigonometric basis $1,\,\cos(x),\,\sin(x),\, \dots$ (right).}
\end{figure}

The rescalings $C^*_{j,n}=(2\pi/(2n+1))C_{j,n}$, $j=0,\ldots,2n$, were considered in~\cite{Roth1} in the context of \textit{Computer Aided Geometric Design (CAGD)}. It was shown therein to actually form a basis for the vector space of trigonometric polynomials (of order at most $n \geq 1$) given by 
\[\mathcal{V}_n=\text{span}\{1,\cos u, \sin u,\ldots,\cos nu, \sin nu\}.\] One important property of these rescalings to the CAGD community is that the resulting basis forms a partition of unity, meaning that $\sum_{j=0}^{2n}C^*_{j,n}(u)=1,$ for all $u \in \mathbb{R}$. The function $\omega_n=2\pi C_{0,n}$ is the so-called 
\textit{De la Vallée Poussin} kernel which has been studied by~\cite{PS58} and $C_{0,n}$ has also been refered to as Cartwright's power of cosine distribution \cite{Cartwright1963}.

We argue here that $\mathcal{B}_n$ provides an interesting model for densities of circular random variables, representing an angle or located on the circumference of a circle. Here is a formal definition of the \textit{angular domain} on which we work.

\label{angular_domain}

Circular random variables take their values on a circle $\s^1$, which we identify to the real line modulo $2\pi$. 
We therefore write $\s^1 = \mathbb{R} \;(\text{mod}\; 2\pi )$, so that $\s^1$ consists of equivalence classes $\{x + 2\pi k : k \in \mathbb{Z}\}$ and is represented by any half-open interval of length $2\pi$. In the following, we do not distinguish equivalence classes from their representatives. We endow $\s^1$ with the \emph{angular distance} $d$ defined as $d_{\s^1}(u, v) = \min_{k \in \mathbb{Z}}| u - v + 2\pi k |$. By the embedding $\theta \mapsto e^{i\theta}$ of $\mathbb{S}^1$ as the unit circle of the complex plane $\mathbb{C}$, the angular distance $d_{\s^1}$ becomes the arc length distance. For instance, an interval $[a,b)\subset \s^1$, $b-a < 2\pi$, can be viewed as an arc of length $b-a$ on the unit circle.

The following result gives elementary properties of the distributions corresponding to the densities in $\mathcal{B}_n$.

\begin{theorem}
The random variables on $\s^1$ given by 
$U_j= U+\frac{2\pi j}{2n+1}$, $j=0,\ldots,2n$, where $U=(1-2V)\cos^{-1}(1-2W)$, with $V$ and $W$ independently distributed, 
$V\sim Ber(1/2)$ and $W\sim Beta(1/2,1/2+n)$, have~\eqref{eq:basis} as densities. Furthermore, by letting $Z_j = 
e^{i U_j}$ be the corresponding random variable on the unit circle of $\mathbb{C}$, we have
\begin{equation}
\label{eq:moments}
\expec(Z_j^p)=
\begin{cases}
\frac{\binom{2n}{n-p}}{\binom{2n}{n}}e^{i \frac{2\pi jp}{2n+1}}, & \text{if } p\in \{-n,\ldots,n\},\\
0 & \text{if } p\in \ints \setminus \{-n,\ldots,n\}.
\end{cases}
\end{equation} 
\end{theorem}

\begin{proof}
The first part is a straightforward application of the change of variables formula. For the integer moments, we have 
the equality
$\expec(Z_j^p)=e^{i\frac{2\pi jp}{2n+1}} \expec\left(Z_0^p\right)$. Using the identity 
\begin{equation}
\label{eq:trigident}
 C_{0,n}(u)=\frac{2^{2n}}{2\pi \binom{2n}{n}}\cos^{2n}(u/2), \quad u \in [0,2\pi),
\end{equation}
and letting $S\sim \mathcal{U}(\s^1)$, we find 
\[
\expec\left(Z_0^p\right)=\frac{1}{\binom{2n}{n}}\sum_{k=0}^{2n}\binom{2n}{k}\expec(e^{-i(n-k-p)S})
=
\begin{cases}
\frac{\binom{2n}{n-p}}{\binom{2n}{n}}, & \text{if } p\in \{-n,\ldots,n\},\\
0 & \text{if } p\in \ints \setminus \{-n,\ldots,n\}.
\end{cases}
\]
\end{proof}
The above integer moments~\eqref{eq:moments} are also known as the Fourier coefficients in \citet[p. 631]{Feller2} and as trigonometric moments in the directional statistics jargon, see for instance \cite{MJ00}, \cite{RS01} and recently \cite{CLB12}. From the result for $p=1$, we get that the mean direction of the $j^{th}$ component is $e^{i \frac{2\pi jp}{2n+1}}$ with the so-called \textit{circular variance} equal to $1/(n+1)$.

\subsection{The circular density model}
Let $\Delta_{2n}$ be the $2n$-dimensional simplex $\Delta_{2n}=\{(c_0,\ldots,c_{2n})\in [0,1]^{2n+1} \,:\, c_0+\cdots 
+c_{2n}=1\}$. Our model consists in mixtures of the form
\begin{equation}\label{eq:Cn}
	C_{n}(u;c_0,\ldots, c_{2n})=\sum_{j=0}^{2n}c_jC_{j,n}(u), \quad u \in \reals,
\end{equation}
with $(c_0,\ldots, c_{2n})\in \Delta_{2n}$, and $n \geq 0$. Let $\mathcal{C}_n$, $n \geq 0$, represent the 
set of mixtures obtained this way; our model is therefore 
\begin{equation}\label{eq:model_C}
	\mathcal{C}=\bigcup_{n\geq 0}\mathcal{C}_n.
\end{equation}

We now give a 
characterization of the model in terms of trigonometric polynomials. We use the following \textit{degree elevation} 
lemma, which is a reformulation of~\citet[Theorem 6]{Roth1}.
\begin{lemma}[Degree elevation formula]
Each $C_{j,n} \in \mathcal{B}_n$ given by~\eqref{eq:basis} can be expressed as
\begin{equation}
\label{eq:degelevc}
C_{j,n}(u)=\sum_{\ell=0}^{2(n+r)}d_{j,\ell}^{n,r}C_{\ell,n+r}(u), 
\end{equation}
with
\begin{equation}
\label{eq:degelevd}
d_{j,\ell}^{n,r}=\frac{1}{2(n+r)+1}\left\{1+\frac{2\binom{2(n+r)}{n+r}}{\binom{2n}{n}} 
\sum_{k=0}^{n-1}\frac{\binom{2n}{k}}{\binom{2(n+r)}{k+r}}\cos\left(\tfrac{2(n-k)\pi \ell}{2(n+r)+1}-\tfrac{2(n-k)\pi j}{2n+1}\right)\right\},
\end{equation}
for $\ell\in \{0,1,\ldots,2(n+r)\}$, and $r\geq 0$.
\end{lemma}
To give the characterization, let $\mathcal{D}_n\subset \mathcal{V}_n$ be the subset of trigonometric polynomial densities (of order at most $n \geq 1$), and let $\mathcal{D}_n^+\subset \mathcal{D}_n$ be the positive ones.

\begin{theorem}[Characterization]\label{characterization}
We have $\mathcal{C}=\bigcup_{n\geq 0}\{\mathcal{B}_n\cup \mathcal{D}_n^+\}$.
\end{theorem}
 
\begin{proof}
If $C_{n} \in \mathcal{C}_n\cap\mathcal{B}_n^c$, then we have $C_{n}(u)>0$ for all $u$, and this shows $\mathcal{C}\subset\bigcup_{n\geq 0}\{\mathcal{B}_n\cup \mathcal{D}_n^+\}$. For the converse inclusion, let $C_{n} \in \mathcal{D}_n^+$, be a positive trigonometric polynomial density, that is, $C_{n}(u)= \sum_{j=0}^{2n}c_j^nC_{j,n}(u)>0$, for all $u \in \s^1$, with $\sum_{j=0}^{2n}c_j^n=1$. Some of the $c_j^n$'s may be negative here. However, by the degree elevation lemma we have
\begin{equation*}
 C_{n}(u)= \sum_{\ell=0}^{2(n+r)}\left\{\sum_{j=0}^{2n}c_j^nd_{j,\ell}^{n,r}\right\}C_{\ell,n+r}(u),  
\end{equation*}
with $d_{j,\ell}^{n,r}$ given by~\eqref{eq:degelevd}. The resulting coefficients $c_{\ell}^{n+r}=\sum_{j=0}^{2n}c_j^nd_{j,\ell}^{n,r}$ also have the property $\sum_{\ell=0}^{2(n+r)}c_{\ell}^{n+r}=1$, and so it remains to show that there is some $r\geq 0$ such that $c_{\ell}^{n+r}\geq 0$, for every $\ell=0,\ldots,2(n+r)$. To see this, use~\eqref{eq:trigident} and the binomial identity to write
\begin{equation*}
C_{n}\left(\tfrac{2\pi \ell}{2(n+r)+1}\right)=\frac{1}{2\pi}\left\{1+\frac{2}{\binom{2n}{n}}\sum_{k=0}^{n-1} \binom{2n}{k} 
  \sum_{j=0}^{2n}c_j^n\cos\left(\tfrac{2(n-k)\pi 
\ell}{2(n+r)+1}-\tfrac{2(n-k)\pi j}{2n+1}\right)\right\}.
\end{equation*}
After some manipulations, and using the fact that $k\mapsto \binom{2(n+r)}{k+r}$ is increasing on $\{0,\ldots,n-1\}$, we find
\begin{align*}
 \left|\tfrac{2(n+r)+1}{2\pi}c_{\ell}^{n+r}-C_{n}\left(\tfrac{2\pi \ell}{2(n+r)+1}\right)\right| &\le
\alpha_1(n)\left(\sum_{k=0}^{n-1} 
\binom{2n}{k}\left|\frac{\binom{2(n+r)}{n+r}}{\binom{2(n+r)}{k+r}}
-1\right|\right)\\
&\le \alpha_2(n)\left(\frac{\binom{2(n+r)}{n+r}}{
\binom{2(n+r)}{r}}-1\right),
\end{align*}
where $\alpha_1(n),\alpha_2(n)>0$. A final calculation shows that
\[
\frac{\binom{2(n+r)}{n+r}}{\binom{2(n+r)}{r}}-1=\frac{(2n+r)(2n+r-1)\cdots (n+r+1)}{(n+r)(n+r-1)\cdots (r+1)}-1\leq 
(1+n/r)^n-1.
\]
Since $C_{n} \in \mathcal{D}_n^+$ is positive by assumption, this shows that for large enough $r$, we have $c_{\ell}^{n+r}> 0$, for every $\ell=0,\ldots,2(n+r)$, and therefore $C_{n}\in \mathcal{C}$.
\end{proof}

As mentioned in the introduction, a criticism made by \cite{FJS08} concerning the nonnegative trigonometric polynomials proposed by \cite{FD04} and \cite{FD07} is that “approximating a function (using nonnegative trigonometric polynomials) often results in a wiggly approximation, unlikely to be useful in most real applications”. 

In the following, we define the notion of \textit{cyclic variations} to formalize “wiggliness” and show that it can be controlled using our basis. 

One way of quantifying “wiggliness” was discussed by~\cite{PS58} via the cyclic variations. For a finite sequence $x=(x_1,\ldots, x_m)$, $m\geq 2$, denote by $v(x)$ the number of sign changes (from positive to negative or vice versa) in the terms of the sequence. Denote by $\vcirc(x)=v(x_i,x_{i+1},\ldots,x_m,x_1,x_2,\ldots,x_{i-1},x_i)$, $x_i\neq 0$, the \emph{cyclic variation} of the sequence, with $\vcirc(x) = 0$ if $x = 0$. This is well defined because $\vcirc$ does not depend on the particular index $i$ such that $x_i \not = 0$. Notice that the value of $\vcirc$ is always an even number not exceeding $m$. The sequence $x$ is said to be \textit{periodically unimodal} if $\vcirc(\dcirc x)=2$, where $\dcirc x = (x_2-x_1,\ldots,x_{m}-x_{m-1}, x_{1} - x_{m})$. For a function $f : \s^1 \rightarrow \mathbb{R}$, we make use of the notation
\[
\vcirc(f)=\sup\{\vcirc(f(x_i)_{i=1}^m): 0\leq x_1<x_2<\cdots<x_m<2\pi, \, m\geq 2\},
\]
and $Z(f)=\#\{x\in [0, 2\pi) : f(x)=0\}$. Similarly to the discrete case, such a function $f$ is said to be \textit{periodically unimodal}, also called \textit{periodically monotone} by~\cite{PS58}, if $\vcirc(f')=2$, provided $f'$ exists (a more general definition without the differentiability assumption is given in the latter paper but is not needed in our case). 

We have the following results.
\begin{theorem}
\label{thm:var}
For $C_{n}=\sum_{j=0}^{2n}c_jC_{j,n}\in \mathcal{C}_n$, let $c=(c_0,\ldots,c_{2n})\in \Delta_{2n}$. We have 
\begin{enumerate}[(i)]
 \item 
\[
\vcirc(C_{n}-\alpha)\leq Z(C_{n}-\alpha)\leq \vcirc\left(\frac{2n+1}{2\pi}c-\alpha\right), \quad \text{for all 
} \alpha\geq 0.
\]
\item A bound for the total variation of $C_{n}$ is given by
\[\tv(C_{n}):=\int_0^{2\pi}|C_{n}'(u)| \, du \leq \frac{2n+1}{2\pi}\sum_{j=0}^{2n}|c_{j+1}-c_j| \leq (2n+1)/\pi,\]
where $c_{2n+1}=c_0$.
\item If $c=(c_0,\ldots,c_{2n})$ is periodically unimodal, then $C_{n}$ is also periodically unimodal.
\end{enumerate}
\end{theorem}
\begin{proof}
The proof of \textit{(i)} follows by~\citet[Lemma 3]{PS58} by noticing that 
\[
	C_{n}(u)-\alpha=\sum_{j=0}^{2n}\left\{\frac{c_j}{2\pi}-\frac{\alpha}{2n+1}\right\}\omega_n\left(u-\frac{2\pi j}{2n+1}\right),\quad u \in \s^1,
\]
with $\omega_n=2\pi \,C_{0,n}$ the \textit{De la Vallée Poussin} kernel. Their result says (in this 
case) that $Z(C_{n}-\alpha)\leq \vcirc\left(c_j/2\pi-\alpha/(2n+1)\right)_{j=0}^{2n}$, which implies \textit{(i)}. 

To show \textit{(ii)}, let $P_n : \s^1 \rightarrow \mathbb{R}$ be the continuous and $2\pi$-periodic, piecewise linear interpolation of the points $(2\pi j/(2n+1),(2n+1)c_j/2\pi)\in \s^1 \times \mathbb{R}$,  $j\in\{0,\ldots, 2n\}$. For definiteness,
\begin{equation}\label{triangular_mixture}
	P_n(u) = \sum_{j=0}^{2n} c_j L_{j}(u), \quad u \in \s^1,
\end{equation}
where $L_{j}(u) = 0 \vee  \frac{2n+1}{2\pi}(1 - \frac{2n+1}{2\pi}d_{\s^1}(u, \frac{2\pi j}{2n+1}))$.
By \textit{(i)} and the Banach Indicatrix Theorem, see \cite{BC09}, we have
\begin{align*}
\tv(C_{n})=\int_0^\infty Z(C_{n}-\alpha) \, d\alpha
&\leq \int_0^\infty\vcirc\left(\frac{2n+1}{2\pi}c-\alpha\right) 
d\alpha,\\
&\leq \int_0^\infty Z(P_n-\alpha) \, d\alpha\\
&= \tv(P_n)=\frac{2n+1}{2\pi}\sum_{j=0}^{2n}|c_{j+1}-c_j|.
\end{align*}
Now a (sharp) bound is easily found for the last sum by $\sum_{j=0}^{2n}|c_{j+1}-c_j|=\|(c_1,\ldots,c_{2n+1})-(c_0,\ldots,c_{2n})\|_1\leq 2$, which leads to the assertion $\tv(C_{n})\leq (2n+1)/\pi$. 

For \textit{(iii)}, we assume $\vcirc(\dcirc c)=2$ and we want to show that $\vcirc(C_{n}')=2$. First, if $\vcirc(C_{n}')=0$ then $C_{n}'$ is either nonnegative or nonpositive. By continuity of $C_{n}'$, we have $0=C_{n}(2\pi)-C_{n}(0)=\int_0^{2\pi}C_{n}'(u)\, du,$ which implies $C_{n}'(u)=0$, for all $u\in [0,2\pi)$, and this gives $c_i=1/(2n+1)$, $i=0,\ldots,2n$. Thus, $\vcirc(C_{n}')=2k$, for some $1\leq k\leq n$. The unit circle $\s^1$ can therefore be partitioned into $2k$ open arcs $A_1,\ldots,A_{2k}$ with $(-1)^j C_{n}$ being nondecreasing on $A_j$, $j=1,\ldots,2k$ and with (anticlockwise) end points $a_1,\ldots,a_{2k}$ (listed in anticlockwise order) being interlaced local minima $\{a_1,a_3\ldots,a_{2k-1}\}$ and maxima $\{a_2,\ldots,a_{2k}\}$ of $C_{n}$. Assume $k>1$ and without loss of generality $a_2\leq a_4$. Let $m=\max\{a_1,a_3\}$. By the monotonicity of $C_{n}$ on each arc, each of which being a connected set (relatively to the topology induced by the angular distance $d$), the Intermediate Value Theorem gives $Z(C_{n}-\alpha)>2$ for all $\alpha \in (m,a_2)$. By the same argument, using the fact  that $\vcirc(\dcirc c)=2$, we obtain 
\[\vcirc\left(\frac{2n+1}{2\pi}c-\alpha\right)=
 \begin{cases}
2, & \text{if } \alpha \in (\min(c),\max(c)),\\
0 & \text{otherwise},
\end{cases}
\]
contradicting \textit{(i)}, and this implies $k=1$.
\end{proof}

\section{Prior specification}
\label{S:Prior}


\subsection{Circular density prior}
 
Our prior $\Pi$ on the space $\bd = \bd(\s^1)$ of bounded circular densities, parametrized by a Dirichlet process $\mathcal{D}$ and a distribution $\rho$ on $\{1, 2, 3, \dots\}$, is induced by the random density
\begin{equation}\label{eq:mixture_prior}
	\sum_{j=0}^{2N} \mathcal{D}(R_{j,N}) C_{j,N},\quad N \sim \rho,
\end{equation}
where $R_{j,n} = \left[\frac{\pi(2j-1)}{2n+1}, \frac{\pi(2j+1)}{2n+1}\right)\subset \s^1$. If $\mathcal{D}$ has a base probability measure $G$ and a concentration parameter $M > 0$, then
\begin{equation} \label{eq:def_Pi}
	\Pi(B) = \sum_{n \geq 0} \rho(n) \Pi_n(B \cap \mathcal{C}_n), \quad B \in \mathcal{B},
\end{equation}
where $\Pi_n = \Pi_{\Delta_{2n}} \circ l_n^{-1}$, $\Pi_{\Delta_{2n}}$ is the Dirichlet distribution of parameters $M G(R_{j,n})$, $j =0,1, \dots 2n$, and where $l_n:\Delta_{2n} \ni (c_0, \dots, c_{2n}) \mapsto \sum_{j=0}^{2n} c_j C_{j,n} \in \mathcal{C}_n$. 

Strong posterior consistency is obtained using Theorem \ref{thm:strong_consistency} of Section \ref{sec:strong_consistency}. The theorem requires the conditional distributions $\Pi_n$ to have full support on $\mathcal{C}_n$,  that $0 < \rho(n) < ce^{-Cn}$ for some $c, C > 0$, and that proper approximation properties of the sieves $\mathcal{C}_n$ are assessed by a sequence $T_n : L^1(\mathbb{M}) \rightarrow L^1(\mathbb{M})$ of linear operators, mapping densities to densities, such that $T_n(\bd) = \mathcal{C}_n \subset \bd$. Here we let $T_n$ be defined by
\begin{equation}\label{eq:Tn_poussin_discretized}
	T_n f = \sum_{j=0}^{2n}\int_{R_{j,n}} f(u)du\, C_{j,n}.
\end{equation}
The only condition of the theorem that is not readily verified is given in the following lemma.

\begin{lemma}\label{lem:approx_circle}
	For every continuous function $f$ on $\s^1$, $\| T_n f - f \|_\infty \rightarrow 0$.
\end{lemma}
\begin{proof}
	We use Lemma \ref{lem:approx}, in the appendix (a result is similar to that of \citet*[Theorem 1.2.1]{Lorentz86}), which gives three sufficient conditions $\textit{(i)}-\textit{(iii)}$ for uniform convergence. We denote $d_{\s^1}(u, R_{j,n}) = \inf_{v \in R_{j,n}} d(u,v)$, and $\diam(R_{j,n})=\sup_{u,v \in R_{j,n}} d_{\s^1}(u,v)$. Here $\textit{(i)}$ is immediate by $\diam(R_{j,n})=2\pi/(2n+1)$, $j=0,\ldots,2n$, and $\textit{(iii)}$ follows from the partition of unity property of $\frac{2\pi}{2n+1}C_{j,n}$. Assumption $\textit{(ii)}$ follows since $C_{0,n}$ is unimodal with mode at $0$, and $d_{\s^1}(u, R_{j,n}) \geq \delta>0$ implies
\[
	C_{j,n}(u) = C_{0,n}\left(d_{\s^1}\left(u, \frac{2\pi j}{2n+1}\right)\right) \le C_{0,n}\left(d_{\s^1}\left(u, R_{j,n}\right)\right) \le C_{0,n}(\delta),
\]
therefore $\sum_{j: d_{\s^1}(u, R_{j,n}) \ge \delta} \frac{2\pi}{2n+1} C_{j,n}(u) \le 2\pi C_{0,n}(\delta) \rightarrow 0$, $n\to \infty$, uniformly over $u\in \s^1$.
\end{proof}

The prior may be interpreted similarly as the Bernstein-Dirichlet prior of \cite{Petrone1999a}. Conditionally on a fixed $n$, the random histogram $H_n = \frac{2n+1}{2\pi} \sum_{j=0}^{2n} c_{j,n} \indic_{R_{j,n}}$ is immediately understood through the Dirichlet distribution on $(c_{0,n}, \dots, c_{2n,n})$. Since $\sum_{j=0}^{2n} c_{j,n}C_{j,n} = T_n H_n$, the following proposition together with Lemma \ref{lem:approx_circle} shows that the finite mixture \eqref{eq:mixture_prior} may be seen as a smooth, variation diminishing approximation to $H_n$.

\begin{proposition}[Variation diminishing property]\label{cor:var}
For every density $f$ on $\mathbb{S}^1$, continuous on $R_{j,n}$, $j=0, \dots, 2n$, we have 
$\vcirc(T_nf-\alpha)\leq \vcirc(f-\alpha)$ for all $\alpha > 0$.
\end{proposition}
\begin{proof}
This is a straightforward consequence of Theorem~\ref{thm:var} \textit{(i)}. Indeed, by continuity of $f$, the Mean 
Value Theorem says that $\prob_f(R_{j,n})=\frac{2\pi}{2n+1}f(u_j)$, for some $u_j\in R_{j,n}$, $j=0,\ldots,2n$. It 
follows that
\begin{equation*}
\vcirc(T_nf-\alpha)\leq
\vcirc\left((\prob_f(R_{0,n}),\ldots,\prob_f(R_{2n,n}))-\alpha\right)\leq \vcirc\left(f-\alpha\right), \quad \alpha >0.
\end{equation*}
\end{proof}

\subsection{Strong posterior consistency}\label{sec:strong_consistency}

We show the strong posterior consistency of a general class of priors for bounded density spaces on compact metric spaces. These include sieve priors such as \eqref{eq:def_Pi}, as well as a class of Dirichlet process location mixtures (see \S \ref{sec:relationship_DPM}). In contrast with \cite{BD12b}, who also obtained general strong consistency result, we consider a prior specification framework, with a different applicability, that does not require continuity and positivity assumptions on the true density from which observations are made. 

Here, strong consistency on $\mathbb{F}$ means that if $X_1,\ldots, X_n$ are independent random variables and identically distributed according to the probability distribution $\prob_{f_0}$ with density $f_0 \in \bd$, denoted $(X_i)_{i \geq 1} \sim P_{f_0}^{(\infty)}$, then for all $\eps >0$,
\begin{equation}
\Pi\left(\left\{f \in \bd \,:\int |f-f_0|< \varepsilon\right\} \mid (X_i)_{i=1}^n\right) 
\rightarrow 1,\quad \prob_{f_0}^{(\infty)}\text{-a.s.}
\end{equation}

The general framework is the following. Suppose $\bd$ is the space of all bounded densities with respect to some finite measure $\mu$ on a compact metric space $(\ms,d)$. Let $T_n : L^1(\mathbb{M}) \rightarrow L^1(\mathbb{M})$, $n\in \mathbb{N}$, be a sequence of linear operators mapping densities to densities. Consider a model having the form $\mathcal{C} = \cup_{n \geq 0} \mathcal{C}_n$, with $\mathcal{C}_n := T_n(\mathbb{F}) \subset \mathbb{F}$. Let $\mathfrak{B}$ be the Borel $\sigma$-algebra of $\mathbb{F}$ for the $L^1$ metric and let $\mathfrak{B}_n$ be the restriction of $\mathfrak{B}$ to $\mathcal{C}_n$, $n \geq 0$. A prior $\Pi$ on $\bd$ can be specified through priors $\Pi_n$ on $(\mathcal{C}_n, \mathfrak{B}_n)$ and a distribution $\rho$ on $n \in \{0,1,2,\dots\}$ as
\begin{equation}\label{eq:sieve_prior}
	\Pi(B) = \sum_{n \geq 0} \rho(n) \Pi_n(B \cap \mathcal{C}_n), \quad B \in \mathfrak{B}. 
\end{equation}

In Theorem \ref{thm:strong_consistency} below, we give simple conditions on $\Pi_n$, $T_n$ and $\rho$, in this framework, ensuring strong posterior consistency on all of $\mathbb{F}$. The proof is given in the appendix.

\begin{theorem}
\label{thm:strong_consistency}
Let $\bd$, $\Pi_n$, $\Pi$ and $T_n$ be as above. Suppose that $T_n(\bd)\subset \bd$ are of finite dimensions bounded by an increasing sequence $d_n\in \mathbb{N}$, and also that $\|T_nf-f\|_{\infty}\to 0$, $n\to \infty$, for every continuous function $f$ on $\mathbb{M}$. If $0< \rho(n)< c e^{-Cd_n}$, for some $c >0$, $C > 0$ and if $\Pi_n$ has support $T_n(\bd)$, then the posterior distribution of $\Pi$ is strongly consistent on $\bd$.
\end{theorem}

The proof is in Appendix \ref{appendix:proof_adaptative}.

\begin{remark}
The result still holds when the space $\bd$ is constrained such as being some convex subset of bounded densities containing at least one density that is bounded away from zero or a star-shaped subset around such a density (e.g. $\bd$ may be a set of bounded unimodal densities or a set of continuous multivariate copula densities). The precise conditions required on $\bd$ are stated at the beginning of Appendix \ref{appendix:proof_consistency}.
\end{remark}

\subsection{Relationship with Dirichlet Process Mixtures}\label{sec:relationship_DPM}

Here we consider Dirichlet Process location Mixtures on $\bd$ induced by the random density 
\begin{equation}\label{eq:DPM}
	f = \int_\mathbb{M} f(\cdot \mid \mu, n) \mathcal{D}(d\mu),
\end{equation}
where $\{f(\cdot \mid \mu, n) \mid  \mu \in \mathbb{M}\} \subset \bd$ are families of densities, $\mathcal{D}$ is a Dirichlet Process and $n$ follows some distribution $\rho$ on $\{1,2,3, \dots\}$. Our circular density prior \eqref{eq:mixture_prior} can be seen to take the form \eqref{eq:DPM} by letting $f(u \mid \mu, n) = \sum_{j=0}^{2n} \mathbb{I}_{R_{j,n}}(\mu) C_{j,n}(u)$. This point of view is especially useful in view of the Slice Sampler of \cite{Walker2007, Kalli2011} which is tailored to Dirichlet Process Mixtures (DPMs).

Furthermore, Theorem \ref{thm:strong_consistency} may be applied to a class of such DPMs. The idea is the following. In order to describe properties of \eqref{eq:DPM}, consider the linear operators $T_n$, $n \in \mathbb{N}$, which maps a probability measure $P$ on $\mathbb{M}$ to the density
\begin{equation}\label{eq:operator}
		T_n P = \int_\mathbb{M} f(\cdot \mid \mu, n) P(d\mu).
\end{equation}
If $P$ has some continuous density $p$, then it is natural to require that 
$
	\|T_n P - p\|_\infty \xrightarrow{n \rightarrow \infty} 0
$ 
(see e.g. assumption A2 in \cite{BD12b}).
If also the image under $T_n$ of all absolutely continuous probability measures is a finite dimensional space, then Theorem \ref{thm:strong_consistency} can be applied to ensure strong posterior consistency.

For instance, we can let 
\begin{equation}\label{eq:prior_pc}
    f(u \mid \mu, n) = C_{0,n}(u - \mu)
\end{equation} to obtain a Dirichlet process mixture over a continuous range of locations. The associated operator $T_n$ defined by \eqref{eq:operator}, when seen as acting on probability densities, is the De la Vallée Poussin mean of \cite{PS58}. Now for any density $f$ on $\mathbb{S}^1$, $T_n f$ is a trigonometric polynomial of degree $n$ \citep{PS58}. Hence the dimension of $T_n(\mathbb{F})$ is bounded above by $2n+1$. Following general theory about integral operators \citep{DL93}, it is straightforward to verify that $\|T_n f -f\|_\infty\rightarrow 0$ for all continuous $f$. Theorem \ref{thm:strong_consistency} is therefore immediately applied to obtain strong posterior consistency. 

In Section \ref{sec:comparison}, a prior of the type \eqref{eq:DPM} with densities given by \eqref{eq:prior_pc} is compared to our circular density prior \eqref{eq:mixture_prior}. Both yield very similar posterior mean estimates in our examples.

\subsection{Adaptative convergence rates}\label{sec:adaptative}

It is interesting to note that the framework of Section \ref{sec:strong_consistency} may be precised as to obtain adaptative convergence rates on classes of smooth densities, similarily as in \cite{Kruijer2008, Shen2015}. Again, the posterior convergence result is stated in some generality as to be easily applicable to other problems of similar nature. 

Here we write $a_n \asymp b_n$ if there are positive constants $A$ and $B$ such that $A b_n \leq a_n \leq B b_n$ for all large $n$. The posterior distribution of $\Pi$ is said to contract around $f_0$ at the rate $\varepsilon_n$ if $(X_i)_{i \geq 1} \sim P_{f_0}^{(\infty)}$ implies that for all large $L > 0$,
\begin{equation}
    \Pi\left(\left\{ f \in \mathbb{F} : H(f_0, f) < L \varepsilon_n \right\} \mid (X_i)_{i=1}^n \right) \rightarrow 1, \quad P_{f_0}^{(\infty)}\text{-a.s.}
\end{equation}
where $H(f_0, f) = \left(\int(\sqrt{f_0} - \sqrt{f})^2\right)^{1/2}$ is the Hellinger distance.

The following assumptions are made on the sequence of operators $T_n$ and on the distribution $\rho$ which induces the prior $\Pi$ defined by \eqref{eq:sieve_prior} with $\Pi_n$ priors on the submodels $T_n(\mathbb{F})$. The proof of Theorem \ref{thm:adaptative_convergence} is in the appendix.

\begin{description}
    \item[\textbf{A1}] The sequence of linear operators $T_n : L^1(\mathbb{M}) \rightarrow L^1(\mathbb{M})$ with $T_n(\mathbb{F}) \subset \mathbb{F}$ maps densities to densities and is such that $\|T_n 1 - 1\|_\infty \rightarrow 0$ for the constant function $1$.
    \item[\textbf{A2}] There exists $d_n \in \mathbb{N}$ an increasing integer sequence with $d_n \geq \dim(T_n(\mathbb{F}))$ and satisfying $d_n \asymp n^d$ for some $d \geq 1$.
    \item[\textbf{A3}] The distribution $\rho$ on $\mathbb{N}$ satisfies $\log(\rho(n)) \asymp -d_n \log(d_n)$.
\end{description}

\begin{theorem}\label{thm:adaptative_convergence}
    Suppose that \textbf{A1}, \textbf{A2} and \textbf{A3} are satisfied. Let $f_0\in \mathbb{F}$ be such that $\|\log f_0\|_\infty < \infty$, $\|T_n f_0 - f_0\|_\infty = \mathcal{O}(n^{-\beta})$ for some $\beta > 0$ and suppose there exists $\kappa > 0$, $\varepsilon_0 > 0$ such that for every large $n \in \mathbb{N}$ and every $0 < \varepsilon < \varepsilon_0 / d_n$,
    \begin{equation}\label{eq:prior_concentration}
        \Pi_n\left(\left\{ f\in T_n(\mathbb{F}) : \| f - T_n f_0 \|_\infty \leq \varepsilon \right\}\right) \geq \left(\varepsilon/d_n\right)^{\kappa d_n}.
    \end{equation}
    Then the posterior distribution of $\Pi$ contracts around $f_0$ at the rate $\varepsilon_n = (n/\log(n))^{-\beta/(2\beta + d)}$.
\end{theorem}

\begin{remark}\label{remark:prior_concentration}
    In order to verify \eqref{eq:prior_concentration}, suppose as in \eqref{eq:Cn} that $$T_n(\mathbb{F}) = \left\{\sum_{j=0}^{d_n} c_{j,n} \phi_{j,n} \mid (c_{j,n})_{j=0}^{d_n} \in \Delta_{d_n}\right\}$$ for some families of basis functions $\{\phi_{j,n}\}_{j=0}^{d_n}$ with $\max_{j} \|\phi_{j,n}\|_\infty \leq C d_n $ for some $C > 0$ that does not depend on $n$. Writing $f = \sum_{j=0}^{d_n} c_{j,n}\phi_{j,n}$ and $T_n f_0 = \sum_{j=0}^{d_n} c_{j,n}^{(0)} \phi_{j,n}$, we find $\|f - T_n f_0\|_\infty \leq C d_n \sum_{j=0}^{d_n} |c_{j,n} - c_{j,n}^{(0)}|$. Now consider a Dirichlet distribution $P$ on the coefficients $(c_{j,n})_{j=0}^{d_n}$ with parameters $(\alpha_{j,n})_{j=0}^{d_n}$ satisfying $\sum_{j=0}^{d_n} \alpha_{j,n} = \alpha$ and $a d_{n}^{-1} < \alpha_{j,n} < b$ for some positive constants $\alpha$, $a$ and $b > 1$ that do not depend on $n$. An application of Lemma A.1 of \cite{Ghosal2001} yields that for every $0 < \varepsilon < \min\{1, 2C/b\}$ and $d_n \geq 2$,
    \begin{align*}
         \Pi_n\left(\left\{ f\in T_n(\mathbb{F}) : \| f - T_n f_0 \|_\infty \leq \varepsilon \right\}\right)
          &\geq P( \{ (c_{j,n})_{j=0}^{d_n} : \sum_{j=0}^{d_n} |c_{j,n} - c_{j,n}^{(0)}| \leq (C d_n)^{-1}\varepsilon \} )\\
          &\geq (\varepsilon/d_n)^{\kappa d_n}
    \end{align*}
    for some $\kappa > 0$ that does not depend on $n$.
\end{remark}

\begin{remark}
    In the case where $f_0 \in T_k(\mathbb{F})$ for some $k \in \mathbb{N}$, the use of $T_nf_0$ to control the approximation error to the sieves may be suboptimal. In this case, it is possible to obtain convergence rates of the order of $(n/\log(n))^{-1/2}$. See for instance \cite{Ghosal2001, Kruijer2008, Barrientos2015}.
\end{remark}

\begin{remark}
    The work in this section shares similarities to \cite{Shen2015} who also obtained general adaptative contraction rates of posterior distributions for a class of random series priors. The reader is refered to \cite{Petrone2010} for a different generalization of the random Bernstein polynomials that is also based on constructive approximation techniques.
\end{remark}

\subsubsection{Application to a circular density prior}

Let us continue the example of Section \ref{sec:relationship_DPM}, where the prior $\Pi$ on the space of all bounded circular densities is a Dirichlet Process location Mixture of $C_{0,n}$ with a distribution $\rho$ on $n\in \mathbb{N}$. The corresponding operator $T_n$ is defined in \eqref{eq:operator} using the densities \eqref{eq:prior_pc}. If $\rho$ is chosen so that $\log(\rho(n)) \asymp -n\log(n)$ and the base distribution of the Dirichlet Process is uniform on $\mathbb{S}^1$ with concentration parameter $\alpha > 0$, Theorem \ref{thm:adaptative_convergence} is easily applied as to obtain the rate of convergence $(n/\log(n))^{-\beta / (2\beta + 2)}$ when $f_0$ is such that $\|\log f_0\|_\infty < \infty$ and satisfies the Hölder continuity condition
$$
    \sup_{x,y \in \mathbb{S}^1} \frac{|f_0(x) -f_0(y)|}{d_{\mathbb{S}^1}(x, y)^\beta} < \infty
$$
for some $\beta \in (0,1]$. Indeed, the operator $T_n$ satisfies the hypothesis \textbf{A1} of Theorem \ref{thm:adaptative_convergence} and \textbf{A2}-\textbf{A3} have already been show to hold. Using Remark \ref{remark:prior_concentration} and the fact that the distribution $\Pi_n$ on the image of $T_n$ corresponds to a Dirichlet distribution on the coefficients of the mixture $\sum_{j=0}^{2n} c_{j,n} C_{j,n}$ with parameters $\alpha_{j,n} = \frac{\alpha}{2n+1}$, we obtain that \eqref{eq:prior_concentration} is satisfied. Furthermore, \cite*[eq. (8.6), Chapter 9]{DL93} shows that $\|T_n f_0 - f_0\|_\infty = \mathcal{O}(\omega_{f_0}(n^{-1/2}))$, where $\omega_{f_0}$ is the modulus of continuity of $f_0$ defined as
$$
\omega_{f_0}(\delta) = \sup \left\{|f_0(x) - f_0(y)|: x, y \in \mathbb{S}^1,\, d_{\mathbb{S}^1}(x, y) < \delta \right\}.
$$
We thus obtain the stated convergence rate $\varepsilon_n = (n/\log(n))^{-\beta/(2\beta + 2)}$ which is, up to log factors, the same as in the case of the random Bernstein polynomial prior \citep{Kruijer2008} for $\beta \in (0,1]$. In the case where $f_0$ is continuously differentiable with $f_0'$ satisfying the Hölder continuity condition with parameter $\alpha \in (0,1]$, then \cite*[eq. (8.6), Chapter 9]{DL93} together with \cite*[eq. (7.13), Chapter 2]{DL93} shows that $\|T_n f_0 - f_0\|_\infty = \mathcal{O}(n^{-(1+\alpha)/2})$. This yields the posterior contraction rate $\varepsilon_n = (n/\log(n))^{-(1+\alpha)/(2(1+\alpha)+2)}$ which is again the same, up to log factors, as for the random Bernstein polynomial prior \citep{Kruijer2008}.
Similar arguments may be used to obtain contraction rates in the case of the De la Vallée Poussin prior \eqref{eq:mixture_prior}.

\section{Comparison of density estimates}\label{sec:comparison}

In this section, we compare density estimates based on the De la Vallée Poussin basis and the nonnegative trigonometric sums of \cite{FD04}. Focus is on the expected Kullback-Leibler and $L^1$ losses in the estimation of target densities exhibiting a range of smoothness, skewness and multimodal characteristics.

\subsection{Nonnegative trigonometric sums}\label{sec:nnts}
Trigonometric polynomials that are probability density functions on the circle can be parameterized by the surface of a complex hypersphere \citep{FD04}. A circular distribution of the corresponding family takes the form
\begin{equation}\label{eq:nnts}
	f(u; c_0, \dots, c_M) = \left\| \sum_{k=0}^M c_k e^{iku} \right\|^2,
\end{equation}
where the coefficients $c_k$ are complex numbers such that $\sum_{k=0}^M \|c_k\|^2 = \frac{1}{2\pi}$.

The parameterization \eqref{eq:nnts} is exploited in \cite{FD04, FD07, FD2010, FD2014a, FD2014b} to model distributions of circular random variables. Circular density estimates from i.i.d. samples are obtained therein by maximum likelihood. Goodness of fit for different degrees $M$ of the trigonometric polynomials is assessed using Akaike's information criterion (AIC) and the Bayesian information criterion (BIC). Recently, \cite{FD2016} considered a uniform prior on the coefficients $c_k$, with respect to hyperspherical surface measure for the Bayesian analysis of circular distributions.

\subsection{Methods}\label{sec:methods}

The following five estimates of circular densities, denoted \textit{pd}, \textit{pc}, \textit{nAIC}, \textit{nBIC} and \textit{fdbayes}, are compared.
\begin{enumerate}
    \item[\textit{pd}:] The posterior mean estimate based on the De la Vallée Poussin prior \eqref{eq:mixture_prior}. This prior is parameterized by a Dirichlet process $\mathcal{D}$ and a probability distribution $\rho$ on $\mathbb{N}$. We chose $\mathcal{D}$ to be centered on the circular uniform distribution with concentration parameter $\alpha = 1$, and we let $\rho(n) \propto e^{-n/5}$.
    \item[\textit{pc}:] The posterior mean estimate based on the Dirichlet process location mixture \eqref{eq:prior_pc}. This prior is also parameterized by a Dirichlet process and a distribution $\rho$ on $\mathbb{N}$. We use the same hyperparameters as above.
    \item[\textit{nAIC}:] The maximum likelihood estimate of \eqref{eq:nnts} where the dimension $M$ is chosen as to minimize Akaike's information criterion.
    \item[\textit{nBIC}:] The maximum likelihood estimate of \eqref{eq:nnts} where the dimension $M$ is chosen as to minimize the Bayesian information criterion.
    \item[\textit{fdbayes}:] The posterior mean estimate based on a uniform hyperspherical distributions on the coefficients $c_k$ of \eqref{eq:nnts} and a uniform prior on $\{0,1,2,\dots, 5\}$ for the dimension $M$. This prior on $M$, uniform on a range $\{0,1,\dots, m\}$ of values, is suggested in \cite{FD2016}. The value of $m=5$, also suggested therein, was chosen as to provide the best performance of this estimator in the comparison of Section \ref{sec:results}.
\end{enumerate}

We assess the quality of a density estimate $f$ using the Kullback-Leibler loss defined by $\int_{\mathbb{S}^1} \log\left(\frac{f_0(u)}{f(u)}\right)f_0(u) du$, where $f_0$ is the target density \citep{Kullback1951}, as well as the $L^1$ loss defined by $\int_{\mathbb{S}^1} |f_0(u) -f(u)|du$. This Kullback-Leibler loss is appropriate in the context of discrimination between density estimates \citep{Hall1987}, while the $L^1$ loss is relevant in view of Theorem \ref{thm:strong_consistency}. Results obtained using the $L^2$ and Hellinger losses were highly similar to those using the $L^1$ loss and we omit their presentation.

\subsubsection{Target densities}
We consider the following two families of target densities to be estimated.
\begin{enumerate}
    \item The \textit{Skewed von Mises} family parameterized by $\alpha \in [0,1]$ and with densities 
    $$v_\alpha(u) \propto (1+\alpha \sin(u+1))\exp(3\alpha \cos(u-\pi)).$$
    \item The family parameterized by $\alpha \in [0, 2\pi)$ and with densities $$w_\alpha(u) \propto \exp(\sin(\cos(2 u) + \sin(3 u) + \alpha)),$$
        which we will refer to as the $w$-family.
\end{enumerate}
The first family was obtained by applying the skewing technique of \cite{Abe2011} to von Mises circular densities and the second family was chosen to showcase multimodal characteristics. This is illustrated in Figure \ref{fig:densities}.

\begin{figure}
    \centering
    \caption{The \textit{Skewed von Mises} family of densities (left panel) and the $w$-family of densities (right panel).}
    \includegraphics[scale=0.6]{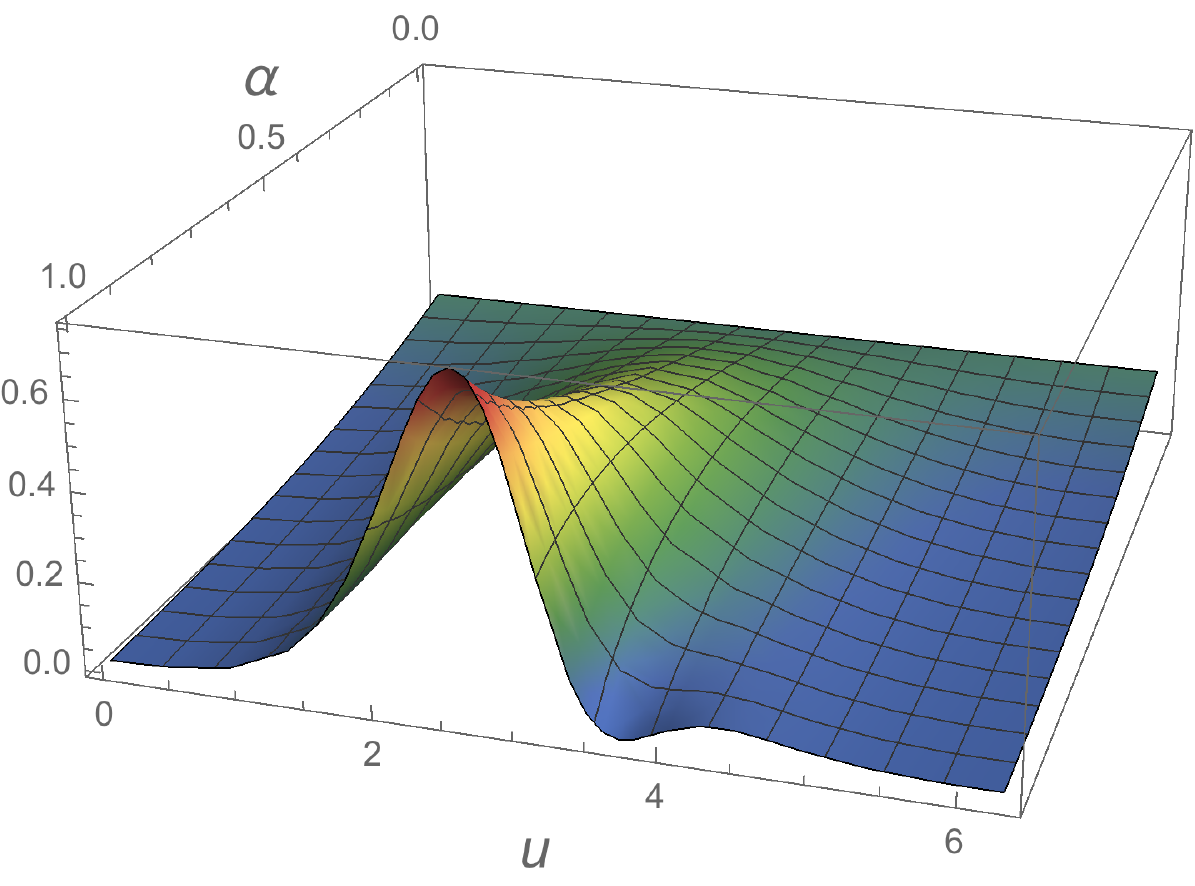}
    \includegraphics[scale=0.6]{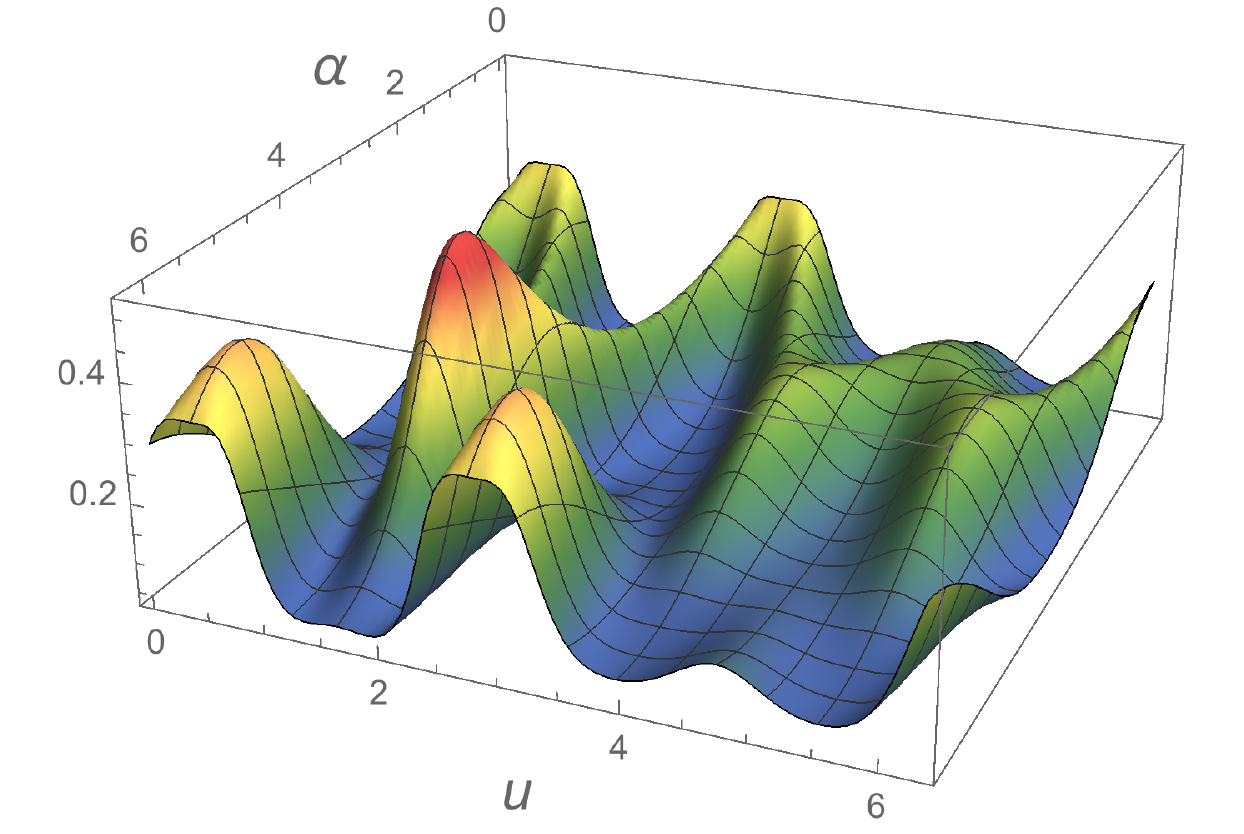}
    \label{fig:densities}
\end{figure}

\subsection{Results}\label{sec:results}

We estimated the mean Kullback-Leibler loss in 1000 repetitions of the estimation of our target densities, for a range of parameter values, using independent samples of sizes $30$ and $100$. The results are shown in Figure \ref{fig:res_skewed_vm} and Figure \ref{fig:res_chaotic}. Bootstrap confidence intervals at the 95\% level are illustrated by vertical bars.

Under the Kullback-Leibler loss, the \textit{nAIC} and \textit{nBIC} estimators are at a considerable disadvantage in the examples considered herein. This is due to their tendency of underestimating probabilities in regions where few samples are observed. An important exception to this, however, is in the use of the the \textit{nBIC} method to estimate a constant densities, since it typically selects $M=0$ or $M=1$ in this case and stays bounded away from zero. 

The Bayesian averaging methods \textit{pc}, \textit{pd} and \textit{fdbayes} are generally more appropriate under the Kullback-Leibler loss and all three are competitive. The \textit{fdbayes} estimator has a poorer performance in the estimation of a spiked unimodal density (\textit{Skewed von Mises} with parameter $\alpha$ near $1$), but improves as the target density approaches being constant.

The \textit{nAIC} estimator improves under a $L^1$ loss. Its increased flexibility over \textit{nBIC} allows to better approach the target in regions of high probability density. The ordering of the estimators is otherwise roughly similar. Under a sample size of size $100$, the different estimators are more clearly distinguished and the \textit{pc} and \textit{pd} estimators provide the best overall performance.

\begin{remark}
    These results show that the De la Vallée Poussin densities provide a viable alternatives to the nonnegative trigonometric sums of \cite{FD04} and that they can be used to adapt techniques developped on the unit interval, such as the random Bernstein polynomials of \cite{Petrone1999a, Petrone2002}, to the topology of the circle. However, it is not our goal to provide best-possible estimators. It would be required to adapt the basis densities as in \cite{Kruijer2008} in order to obtain certain minimax-optimal Hellinger convergence rates. Our theoretical results can also be applied when using different density bases, including for multivariate density estimation, and the shape-preserving properties of the De la Vallée Poussin densities can be used to incorporate prior information.
\end{remark}

\begin{figure}
\caption{Mean Kullback-Leibler losses for the \textit{Skewed von Mises} family $\{v_\alpha\}$ of target densities and different values of the parameter $\alpha$.}
	\includegraphics[width=\linewidth]{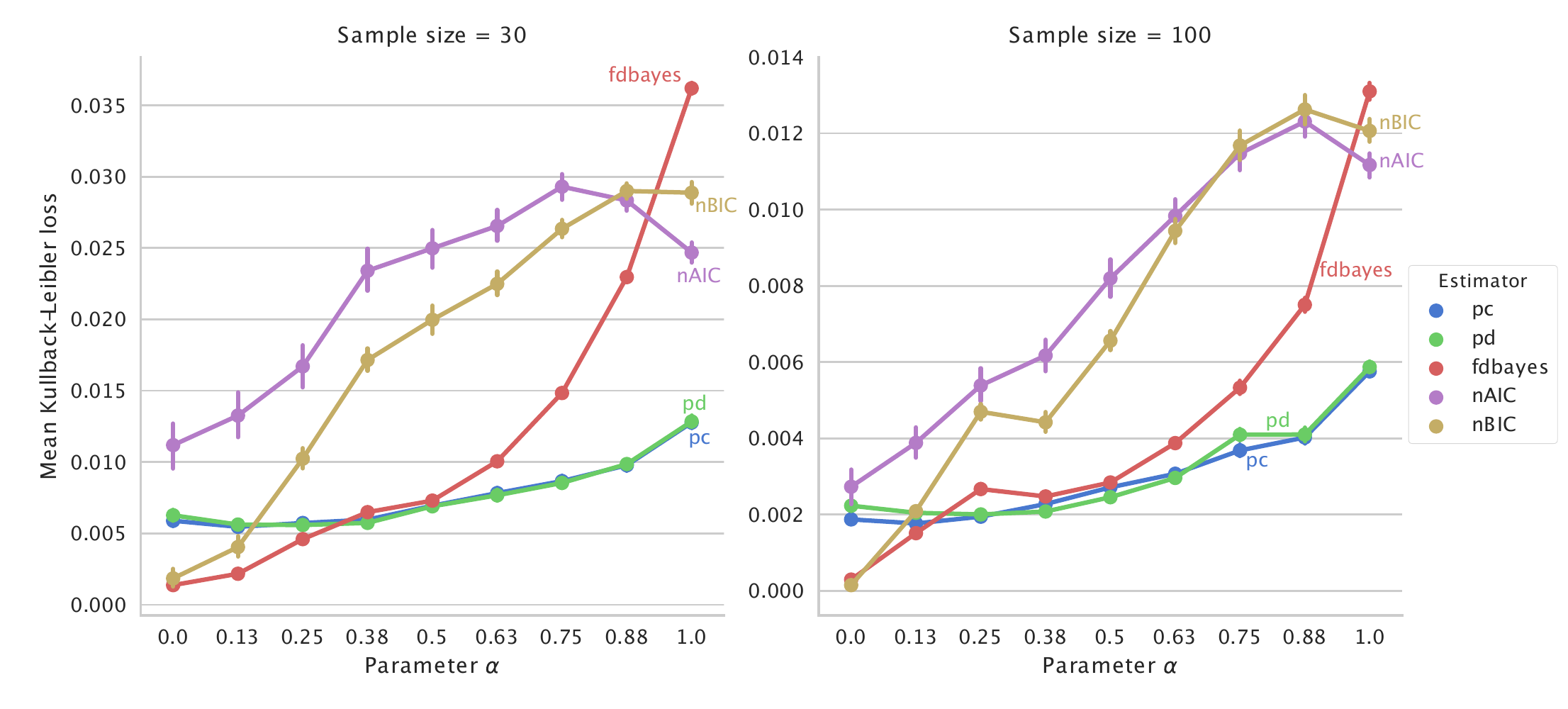}
	\includegraphics[width=\linewidth]{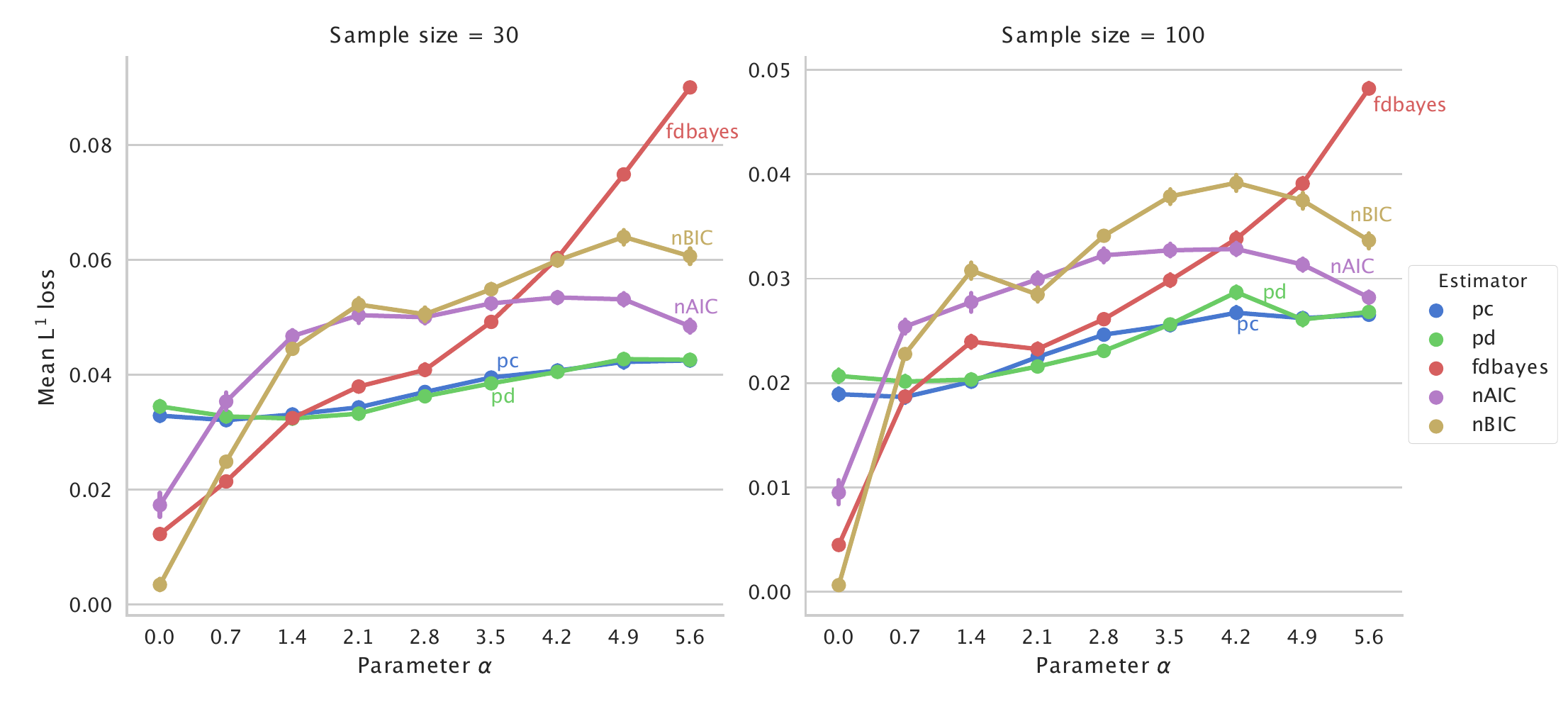}
	\label{fig:res_skewed_vm}
\end{figure}

\begin{figure}
\caption{Mean Kullback-Leibler losses for the $w$-family $\{w_\alpha\}$ of target densities and different values of the parameter $\alpha$.}
	\includegraphics[width=\linewidth]{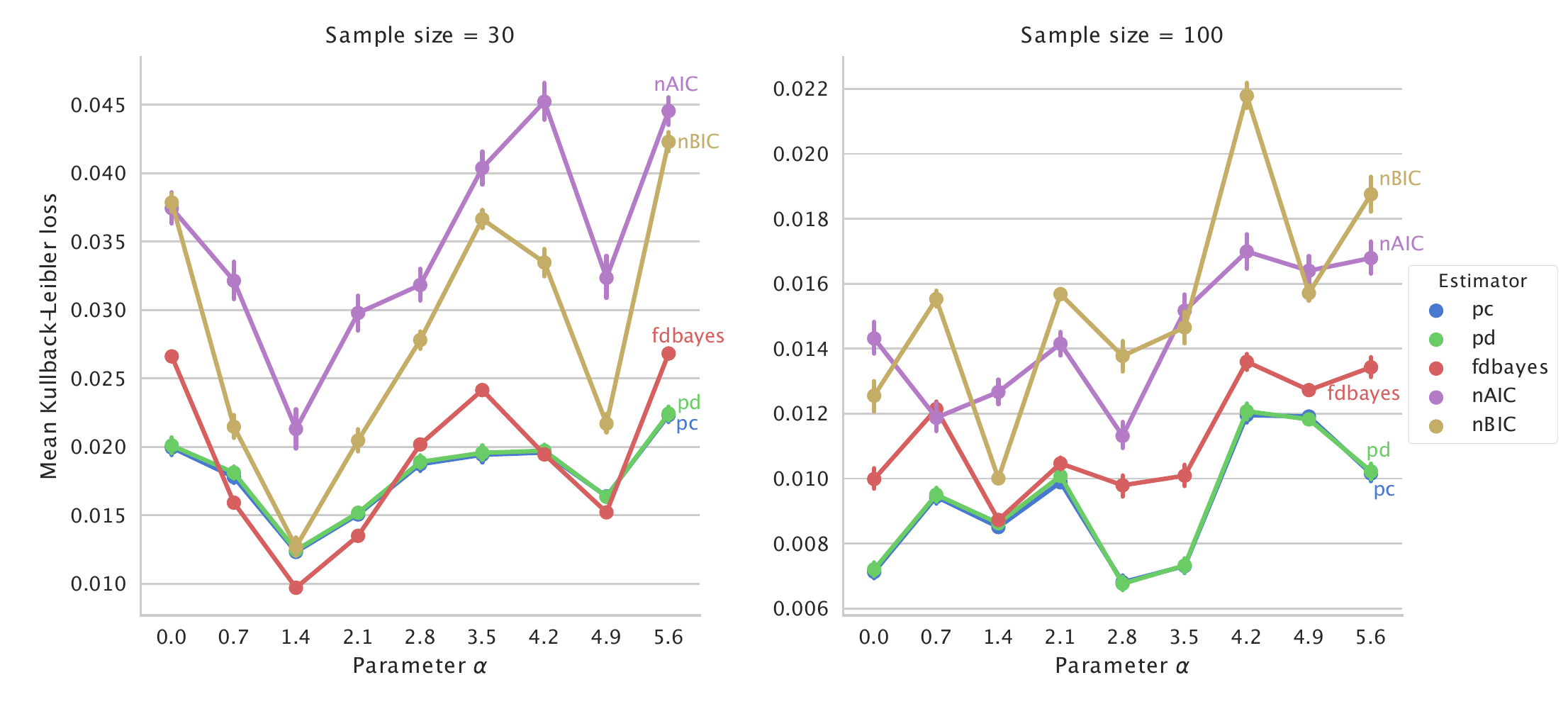}
	\includegraphics[width=\linewidth]{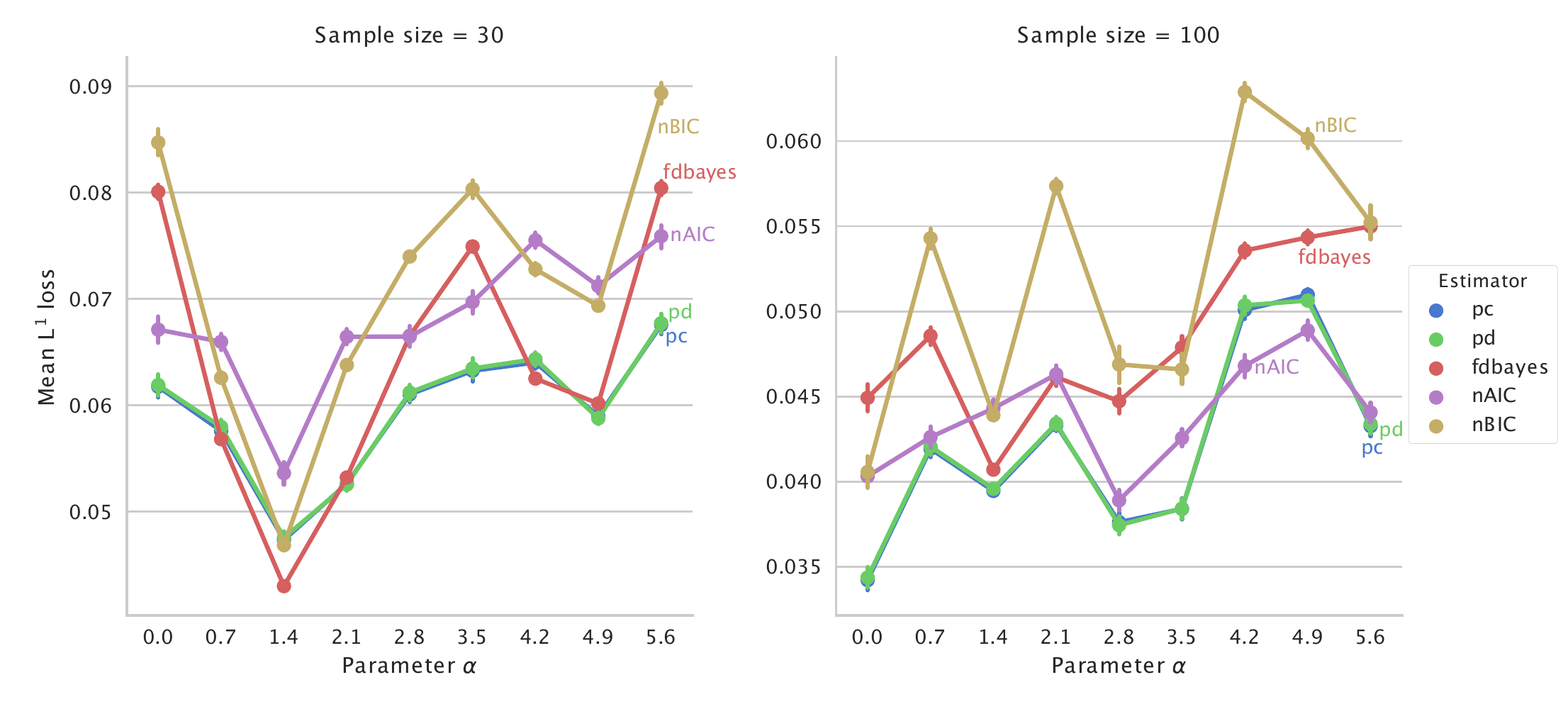}
	\label{fig:res_chaotic}
\end{figure}

\subsection{Implementation summary}\label{sec:implementation}

\begin{figure}
\caption{Examples of density estimates for different targets and sample sizes.}
    \centering
    \includegraphics[width=\linewidth]{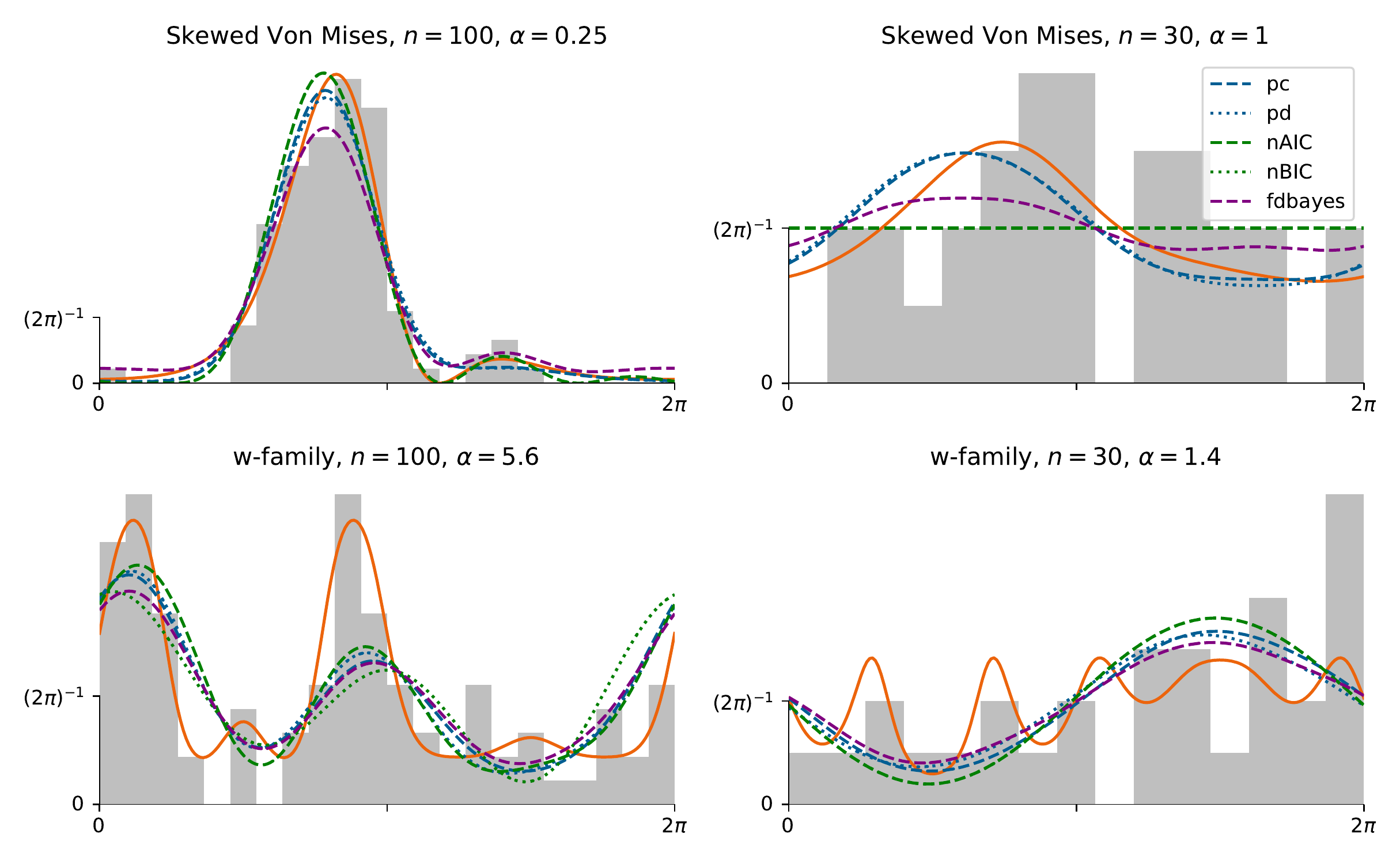}
    \label{fig:est_examples} 
\end{figure}

The \textit{nAIC} and \textit{nBIC} density estimates are obtained using the CircNNTSR R package \citep{JSSv070i06}. Precisely, we ran the function ``nntsmanifoldnewtonestimation'' twice from random starting points provided by ``nntsrandominitial'' and for each degree $M$ of the trigonometric polynomials ranging in $\{0,1,\dots, 7\}$. Density estimates with the best AIC and BIC scores were retrieved.

Posterior means corresponding to the \textit{pc} and \textit{pd} estimates are approximated using the Slice Sampler described in \cite{Kalli2011}. The implementation is straightforward. We ran 80 thousand iterations of the algorithm, of which 20 thousand were treated as burn-in, and sub-sampled down to 20 thousand iterations in order to calculate the posterior mean. Each iteration consisted in the update of every variable in the Slice Sampler following their full conditional distribution. The distribution of the model dimension $n$ was truncated to the range $\{1, 2, 3, \dots, 60\}$.

Posterior means for the \textit{fdbayes} estimates are approximated using a simple independent Metropolis-Hastings algorithm with trans-dimensional moves that naturally exploit the nestedness of the models. We ran the algorithm for a million iterations, treating 100 thousand as burn-in, and sub-sampled down to 20 thousand observations in order to calculate the posterior mean. This large number of iterations was used to ensure convergence across the 7200 different datasets and to compensate for the lower acceptance rate of independent Metropolis-Hastings.

\section{Discussion}
    
We introduced the density basis $C_{j,n}$, $j \in \{0,1,\dots, 2n\}$, of the trigonometric polynomials. It is well suited to mixture modelling in the sense that different characteristics of the mixture density $ f = \sum_{j=0}^{2n} c_{j,n} C_{j,n}$ can be easily related to the vector $c = (c_{0,n}, c_{1,n}, \dots, c_{2n,n})$ of coefficients. For instance, Theorem \ref{thm:var} shows that $f$ is constant if and only if $c$ is constant; that it is periodically unimodal if $c$ is periodically unimodal; and that the range of $f$ is contained between $\frac{2n+1}{2\pi} \min\{c_{j,n}\}_{j=0}^{2n}$ and $\frac{2n+1}{2\pi} \max\{c_{j,n}\}_{j=0}^{2n}$. From the cyclic symmetry of the basis, it also follows that $f$ is symmetric about $0$ if the vector $(c_{n+1,n}, \dots, c_{2n, n}, c_{0,n},c_{1,n}, \dots, c_{n,n})$ is symmetric about its center coefficient $c_{0,n}$. As yet another example, consider the problem of modelling a bivariate angular copula density $g: \mathbb{S}^1 \times \mathbb{S}^1 \rightarrow [0,\infty)$. Using the De la Vallée Poussin basis, we may let $g(u,v) = \sum_{i,j = 0}^{2n} c_{i,j} C_{i,n}(u)C_{j,n}(v)$. The fact that $g$ has constant marginal densities follows if the row sums and column sums of the matrix of coefficients $[c_{i,j}]_{i,j}$ are constant. On the interval $[0,1]$, similar properties of the Bernstein polynomial densities have been exploited for copula modelling and shape constrained regression \citep{Guillotte2012, Chang2007}. The De la Vallée Poussin basis may thus be used to adapt such procedures developed in the unit interval case to the topology of the circle.

\section*{Aknowledgements}

The authors are grateful to the Natural Sciences and Engineering Research Council of Canada (NSERC) for a Discovery grant (S. Guillotte) as well as an Alexander Graham Bell Canada Graduate Scholarship (O. Binette).

\begin{appendices}

\section{Proof of Theorem \ref{thm:strong_consistency}}\label{appendix:proof_consistency}

Let $\bd$ be any space of bounded densities such that for all $f \in \bd$, there 
exists $h\in \bd$ with $\inf_x h(x) > 0$ and 
$\{(1-\alpha)f + \alpha h : 0 < \alpha < 1\} \subset \bd$ 
(the assumption is used only at the end of the proof in \textit{Claim 3}). We also recall the hypothesis $\mathcal{C}_n := T_n(\bd) \subset \mathbb{\bd}$.

\subsection{Some notations}

Let $\|\cdot\|_\infty$ denote the supremum norm, let $\|\cdot\|_1$ denote the 
$L^1$-norm, 
and write $\Lball(f_0, \varepsilon) = \{f \in\bd : \|f-f_0\|_1 < 
\varepsilon\}$, $\eps>0$, for an $L^1$-ball. For a subset $A \subset \bd$ and 
$\delta > 0$, let $N(A,\delta)$ be the minimum number of $L^1$-balls of 
radius $\delta$ and centered in $\bd$ needed to cover $A$. Let 
$\KL(f_0,f)=\int_{\{f_0>0\}} f_0 \log f_0/f \, d\mu$ be the Kullback-Leibler 
divergence between the densities $f_0$ and $f$, and denote $ \KLball(f_0, 
\varepsilon) := \left\{ f \in \bd : \KL(f_0,f) < 
\varepsilon \right\}$. The \emph{Kullback-Leibler support} of $\Pi$ is the set 
of all densities $f_0$ such that
$\Pi ( \KLball(f_0, \varepsilon) ) > 0,$
for all $\varepsilon > 0$. Note that the $\mathfrak{B}$-measurability of 
$\KLball(f_0, \varepsilon)$ is shown in \citet*[Lemma 11]{Barron:1999}.

\subsection{A result of \cite{Xing2009}}

Strong consistency on the Kullback-Leibler support of $\Pi$ is ensured as a particular case of \citet*[Theorem 2]{Xing2009} (see also \cite{walker2004, Lijoi05}) which we state here in the following lemma (their result is stated in terms of the Hellinger distance which is topologically equivalent to the $L^1$-distance). The fact that $\ms$ is a finitely measured compact metric space satisfies the conditions on $\ms$ and $\bd$ stated therein. Therefore, once we show that the lemma applies, all we need is to compute the Kullback-Leibler support. 
\begin{lemma}\label{lem:xing}
	Let $\mathcal{F}_n \subset \bd$, $n \in \mathbb{N}$, be such that $\Pi(\cup_{n} \mathcal{F}_n) = 1$. Suppose there exists $\alpha:(0,1) \rightarrow [0,1)$ such that $\lim_{\delta \rightarrow 0} {\delta}/{(1-\alpha(\delta))} = 0$ and
	\begin{equation}\label{hausdorff_bound}
		\sum_{n=0}^\infty N(\mathcal{F}_n, \delta)^{1-\alpha(\delta)}\Pi(\mathcal{F}_n)^{\alpha(\delta)} < \infty
	\end{equation}
	for every small $\delta > 0$. Then the posterior distribution of $\Pi$ is strongly consistent at every density $f_0$ of its Kullback-Leibler support.
\end{lemma}

\subsection{Application of the lemma}\label{sec:application_of_the_lemma}

Denote 
$\overline{\mathcal{C}_n}$ the $L^1$-closure of $\mathcal{C}_n=T_n(\bd)$ in 
$\bd$. We apply Lemma \ref{lem:xing} with the disjoint 
$\mathfrak{B}$-measurable sets $\mathcal{F}_n = 
\overline{\mathcal{C}_n}\bigcap_{0\leq k < n}\overline{\mathcal{C}_k}^c$, so 
that $\Pi(\cup_n \mathcal{F}_n) = \Pi(\cup_n \overline{\mathcal{C}_n}) = 1$ and 
$\Pi(\mathcal{F}_n) = \sum_{k \geq 0} \rho(k) \Pi_{k}(\mathcal{F}_n \cap 
\mathcal{C}_k) \le \sum_{k \geq n} \rho(k)$. Let $d_k$ be the strictly increasing integer sequence bounding $\dim (\mathcal{F}_k)$ and such that $\rho(k) < c e^{-C d_k}$, so that we find $\sum_{k \geq n} \rho(k) < c \sum_{k \geq n} e^{-Cd_k} \le c \sum_{k \geq d_n}e^{-Ck}\varpropto e^{-Cd_n}$. Moreover, from Lemma 1 of \cite{Lorentz:1966}, 
$\mathcal{F}_n$ being of dimension at most $d_n$ and contained in an $L^1$-ball of radius 2, we have $N(\mathcal{F}_n, \delta)\le (6/\delta)^{d_n}$. It follows that
\begin{align*}
    \sum_{n=0}^\infty N(\mathcal{F}_n, \delta)^{1-\alpha(\delta)}\Pi(\mathcal{F}_n)^{\alpha(\delta)} \leq D \sum_{n = 0}^{\infty} \exp\left( -d_n\left\{(1-\alpha(\delta))\log(\delta/6)+\alpha(\delta) C\right\} \right)
\end{align*}
for some constant $D > 0$.
Now let $\alpha(\delta) = (1-\delta)^{-\log(\delta)}$, noting that $\lim_{\delta \rightarrow 0} \alpha(\delta) = 1$ and 
$$
\alpha'(\delta) = \alpha(\delta)\left(\frac{\log(\delta)}{1-\delta} - \frac{\log(1-\delta)}{\delta}\right).
$$
Hence, $\lim_{\delta\rightarrow 0} \delta/(1-\alpha(\delta)) = -\left(\lim_{\delta \rightarrow 0} \alpha'(\delta)\right)^{-1} = 0$.
Furthermore, the series~\eqref{hausdorff_bound} 
converges provided$(1-\alpha(\delta))\log(\delta/6)+\alpha(\delta) C > 0$ for $\delta> 0$ sufficiently small. This is indeed the case since $\lim_{\delta \rightarrow 0} C\alpha(\delta) = C > 0$ and $\lim_{\delta \rightarrow 0} (1-\alpha(\delta))\log(\delta/6) = 0$.

\subsection{The Kullback-Leibler support of $\Pi$}

Let $\KL(\Pi)$ denote Kullback-Leibler support of $\Pi$; we show that $\bd \subset \KL(\Pi)$. The proof is divided in the three following claims.

\medskip
\noindent\textit{\textbf{Claim 1:} For all $f\in L^1(\ms)$ we have 
$\|T_n f - f\|_1 \rightarrow 0$.}

To see this, the fact that $T_n$ maps the densities of $ L^1(\ms)$ to 
densities implies that $f\mapsto T_n f$, $f \in L^1(\ms)$, is monotone 
and we get $\|T_n f\|_1\leq \|T_n |f|\|_1\leq \|f\|_1$, for all $n\geq 0$. 
Take $\eps > 0$, we can find $g$ continuous with $\|f - g\|_1 < \eps/3$; this is because the set of continuous functions on $\ms$ is dense in $L^1(\ms)$.
Now by assumption there exists 
$N\geq 0$ such that $\|T_N g - g\|_\infty < \varepsilon/(3\mu(\ms))$, 
and we get $\|T_N f- f\|_1 \leq \|T_N (f - g)\|_1 + \|T_N g - g\|_1 + \|g - 
f\|_1 < \eps$. 

Now let $\bd^+$ be the densities in $\bd$ which are bounded away from zero. 

\medskip
\noindent\textit{\textbf{Claim 2:} $\bd^+ \subset\KL(\Pi)$}. 

We show that for all $f_1 \in \bd^+$, and for all $\eps>0$, there exists an 
$N\geq 0$ and $\delta>0$ such that $\Lball(T_N f_1, 
\delta) \cap \mathcal{C}_N\subset \KLball(f_1, \varepsilon)$. The result will 
then follow from 
\[
\Pi(\KLball(f_1, \varepsilon))=\sum_{k \geq 0} \rho(k) \Pi_{k}(\KLball(f_1, 
\varepsilon) \cap 
\mathcal{C}_k) \geq 
\rho(N)\Pi_{N}(\Lball(T_N f_1, \delta) \cap \mathcal{C}_N) > 0,\] 
since $\rho(N) > 0$ and $\Pi_{N}$ has support $\mathcal{C}_N$. To find such $N$ 
and $\delta$, notice that for all $f\in \bd^+$,
\begin{equation}
 \label{eq:KL}
\KL(f_1,f)\leq 
\|f_1/f\|_\infty\|f_1 -f\|_1\leq \|f_1/f\|_\infty(\|f_1 -T_nf_1\|_1 + 
\|T_nf_1-f\|_1).
\end{equation}
Now put $0<\inf_{x\in 
\ms} f_1(x)=:m\leq M:=\sup_{x\in 
\ms} f_1(x)$. By the first claim, there exists $N \geq 0$ 
such that $\| T_n f_1 - f_1 \|_1 < \frac{m}{8 M}\eps$, for all $n\geq N$. 
Furthermore, since $f\mapsto T_nf$ is monotone and since $\|T_n m-m\|_{\infty} 
\to 0$, we can assume $N$ is large enough so that we also have $\inf_{x\in\ms} 
T_N f_1(x) \geq \inf_{x\in\ms} T_N m(x)  \geq m/2$. Since 
$\mathcal{C}_N=T_N(\bd)\subset \bd$ and is finite 
dimensional, $\|\cdot\|_{\infty}$ is finite and equivalent to $\|\cdot\|_1$ on 
$\mathcal{C}_N$ and we can find $0 < \delta < \frac{m}{8 M}\eps$ such that 
$\Lball(T_N f_1, \delta) \cap \mathcal{C}_N \subset \Sball(T_N f_1, m/4)\cap 
\mathcal{C}_N$. Now for any $f \in \Lball(T_N 
f_1, \delta) \cap \mathcal{C}_N$, the quantity $\|f_1 / f\|_\infty 
\le 4M/m$, so that by plugging $N$ in~\eqref{eq:KL} we get $\KL(f_1,f)<\eps$.

\medskip
\noindent\textit{\textbf{Claim 3:} $\bd\setminus \bd^+ \subset\KL(\Pi)$}.

Let $f_0\in \bd\setminus \bd^+$ and let $0 < \varepsilon < 6$. By assumption 
there is an $h \in \bd^+$ such that 
$\{(1-\alpha) f_0 + \alpha h : 0 < \alpha < 1\} \subset \bd$. Now take $f_1 = 
\frac{f_0 + \gamma h}{1+\gamma} \in 
\bd^+$, with $\gamma = \eps/6$, so $f_0<(1+\gamma)f_1$. We use the 
following result from \citet*[Lemma 5.1]{Ghosal:1999a}. 
\begin{lemma}
If $f_0$ and $f_1$ are densities with $f_0 
\le C f_1$, for some $C\geq 1$, then for any density $f$,
\[
\KL(f_0, f) \le (C+1)\log C + C\left[\KL(f_1, f) + \sqrt{\KL(f_1, f)} 
\right].
\]
\end{lemma}
Here $(2+\gamma)\log(1+\gamma)<\eps/2$. By the second claim and the above 
lemma, there exists $\delta > 0$ and $N\geq 0$ such that for $f \in \Lball(T_N 
f_1, \delta) \cap \mathcal{C}_N$, we have $\KL(f_0,f)< \eps$.

\section{Proof of Theorem \ref{thm:adaptative_convergence}}\label{appendix:proof_adaptative}

We apply a particular case of \citep[Theorem 1]{Xing2011} which is stated in the following lemma. Here $H(f_0, f)^2 = \int \left(\sqrt{f} - \sqrt{f_0}\right)^2 d\mu$ is the squared Hellinger distance and $N(\varepsilon, \mathcal{F}; H)$ is the covering number of $\mathcal{F}$ with respect to the Hellinger distance: it is the minimum number of Hellinger balls of radius $\varepsilon$ necessary to cover $\mathcal{F}$.

\begin{lemma}[\cite{Xing2011}]
    Let $\varepsilon_n$ and $\tilde \varepsilon_n$ be positive sequences such that $n\min\{\varepsilon_n^2, \tilde \varepsilon_n^2\} \rightarrow \infty$ as $n \rightarrow \infty$. Suppose there exists subsets $\mathcal{F}_j$, $j \in \mathbb{N}$, of $\mathbb{F}$ with $\Pi(\cup_j \mathcal{F}_j) = 1$ and constants $c_1 > 0$, $c_2 > 0$, $0 \leq \alpha < 1$ such that 
    \begin{equation}\label{eq:Xing_c1}
        \sum_{n=1}^\infty e^{-c_1 n \tilde \varepsilon_n^2} \sum_{j=1}^\infty N(\tilde \varepsilon_n, \mathcal{F}_j; H)^{1-\alpha} \Pi(\mathcal{F}_j)^\alpha < \infty
    \end{equation}
    and
    \begin{equation}\label{eq:Xing_c2}
        \Pi\left(\left\{ f \in \mathbb{F} : H(f_0, f)^2 \|f_0/f\|_\infty^{1/2} \leq \varepsilon_n^2 \right\}\right) \geq e^{- n\varepsilon_n^2 c_2}
    \end{equation}
    for all large $n$.
    Then the posterior distribution of $\Pi$ contracts around $f_0$ at the rate $\max \{\varepsilon_n, \tilde \varepsilon_n\}$.
\end{lemma}
Here we let $\tilde \varepsilon_n = n^{-\gamma}$ for $\gamma$ satisfying $\beta/(2\beta + d) < \gamma < 1/2$, and $\varepsilon_n = (n/\log(n))^{-\beta/(2\beta + d)}$. The two conditions \eqref{eq:Xing_c1} and \eqref{eq:Xing_c2} can be independently verified.

\subsection{Verification of condition \eqref{eq:Xing_c1}}

This follows along the lines of Section 3.1 in \cite{Xing2008_arxiv}.
By assumption \textbf{A3}, there exists a constant $C > 0$ such that $\rho(n) \leq e^{-C d_n \log(d_n)}$. 
As in the proof of Theorem \ref{thm:strong_consistency}, we let $\mathcal{F}_j = \overline{\mathcal{C}_j}\bigcap_{0\leq k < j}\overline{\mathcal{C}_k}^c$ with $\mathcal{C}_j = T_j(\bd)$. Now using \textbf{A2}, $\Pi(\mathcal{F}_j) \leq \sum_{k \geq j} \rho(k) \leq  \sum_{k \geq d_j}e^{-Ck \log (k)}$ is bounded above by $ L e^{-C d_j \log(d_j)}$, $L = 2^C/(2^C -1)$, when $j \geq 2$. 
Since $H(f, g)^2 \leq \int |f-g| \,d\mu$, we have that $N(\tilde \varepsilon_n, \mathcal{F}_j;H) \leq N(\tilde \varepsilon_n^2, \mathcal{F}_j) \leq (6/\tilde \varepsilon_n^2)^{d_j}$ where the last inequality is derived as in Appendix \ref{sec:application_of_the_lemma}. 

Now let $0 \leq \alpha < 1$ be sufficiently close to $1$ so that $ C \alpha (1-2\gamma) \geq 2\gamma(1-\alpha)$. By Lemma \ref{lem:sum_order}, there exists $D > 0$ with $\sum_{j=1}^\infty \left(\frac{j^{C\alpha}}{6^{1-\alpha}n^{2\gamma(1-\alpha)}}\right)^{-j} \leq \exp \left(D n^{2\gamma(1-\alpha)/(C \alpha)} \right)$ for every large $n$. We therefore obtain
    \begin{align*}
        \sum_{j=1}^\infty N(\tilde \varepsilon_n, \mathcal{F}_j;H)^{1-\alpha} \Pi(\mathcal{F}_j)^\alpha
        & \leq L^\alpha\sum_{j=1}^\infty (6 n^{2\gamma})^{d_j(1-\alpha)} e^{-C d_j \log(d_j) \alpha}\\
        & \leq L^\alpha\sum_{j=1}^\infty (6 n^{2\gamma})^{j(1-\alpha)} e^{-C j \log(j) \alpha}\\
        & = L^\alpha \sum_{j=1}^{\infty} \left( \frac{j^{C\alpha}}{6^{1-\alpha}n^{2\gamma(1-\alpha)}} \right)^{-j} \leq L^\alpha \exp \left(D n^{2\gamma(1-\alpha)/(C \alpha)} \right).
    \end{align*}
    Taking $c_1 > D$ and since $(1-2\gamma) \geq 2\gamma(1-\alpha)/(C \alpha)$, it follows that 
    \begin{align*}
        &\sum_{n=1}^\infty e^{-n \tilde \varepsilon_n^2 c_1} \sum_{j=1}^\infty N(\tilde \varepsilon_n, \mathcal{F}_j)^{1-\alpha} \Pi(\mathcal{F}_j)^\alpha\\
        &\leq L^\alpha \sum_{n=1}^\infty \exp\left(D n^{2\gamma(1-\alpha)/(C \alpha)} - c_1 n^{1-2\gamma}\right) < \infty.
    \end{align*}

\subsection{Verification of condition \eqref{eq:Xing_c2}}

This follows along the lines of the proof of Theorem 2.3 in \cite{Ghosal2001} and of the proof of Theorem 2 in \cite{Kruijer2008}. Again $\varepsilon_n = (n/\log(n))^{-\beta/(2\beta + d)}$ and we let $k_n$ be an integer sequence such that $k_n \asymp \varepsilon_n^{-1/\beta}$. The first step of the proof is to show that for some constant $L_1 > 0$ and for $n$ sufficiently large,
\begin{align}\label{eq:B2_first_step_goal}
    \left\{ f : H(f_0, f)^2 \|f_0/f\|_\infty^{1/2} \leq L_1\varepsilon_n^2\right\}
    \supset  \left\{ f \in T_{k_n}(\mathbb{F}) : \|T_{k_n} f_0 - f\|_\infty \leq \varepsilon_n \right\}.
\end{align}
The probability of the set on the right hand side will then be lower bounded through \eqref{eq:prior_concentration}.

Since $\|\log f_0\|_\infty < \infty$ by assumption, there exists constants $m$, $M$ with $0 < m < f_0 < M$. Furthermore, if $f\in \mathbb{F}$ is such that $\|T_n f_0 - f\|_\infty < \inf T_n f_0$, then
$$
    \|f_0/f\|_\infty \leq \frac{M}{(\inf T_n f_0) - \|T_n f_0 - f\|_\infty}.
$$
By assumption \textbf{A1} and the resulting positivity of $T_n$, $\inf T_n f_0 \geq T_n(m) \rightarrow m$ as $n \rightarrow \infty$. Hence for $n$ sufficiently large that $\inf T_n f_0 > m/2$ and if $\|T_n f_0 - f\|_\infty < m/4$, then
$$
    \|f_0/f\|_\infty \leq \frac{M}{m/2 - \|T_n f_0 - f\|_\infty} \leq 4M/m.
$$
Now, since we are integrating with respect to the finite measure $\mu$, we also have
\begin{align*}
    H(f_0, f)^2 
    & \leq \int \left(\sqrt{f} - \sqrt{f_0}\right)^2\left(1+\sqrt{f/f_0}\right)^2d\mu\\
    & \leq m^{-1}\int (f-f_0)^2d\mu\\
    & \leq m^{-1}\mu(\mathbb{M})\|f-f_0\|_\infty^2.
\end{align*}
Furthermore, $\|f-f_0\|_\infty \leq \|T_{k_n f_0} - f_0\|_\infty + \|T_{k_n}f_0 - f\|_\infty$ with $\|T_{k_n f_0} - f_0\|_\infty = \mathcal{O}(k_n^{-1/\beta})$ and $k_n^{-\beta} \asymp \varepsilon_n$. Therefore, taking $n$ sufficiently large that $\inf T_{k_n} f_0 > m/2$ and $\varepsilon_n \leq m/4$, we have that $\|T_{k_n}f_0 - f\|_\infty \leq \varepsilon_n$ implies
$$
    H(f_0, f) \|f_0/f\|_\infty^{1/4} \leq L_2 (k_n^{-\beta} + \varepsilon_n) \leq L_3 \varepsilon_n
$$
for some constants $L_2$ and $L_3$. This proves \eqref{eq:B2_first_step_goal}.

Now for $n$ sufficiently large, we have $\varepsilon_n^{1+d/\beta} \leq \varepsilon_n$ and $\varepsilon_n^{1+d/\beta} \leq \varepsilon_0 / d_{k_n}$, where $\varepsilon_0$ is a fixed constant in Theorem \ref{thm:adaptative_convergence}. Hence using \eqref{eq:prior_concentration} we find
\begin{align*}
    \Pi\left(\left\{ f \in T_{k_n}(\mathbb{F}) : \|T_{k_n} f_0 - f\|_\infty \leq \varepsilon_n \right\}\right)
    &\geq \Pi\left(\left\{ f \in T_{k_n}(\mathbb{F}) : \|T_{k_n} f_0 - f\|_\infty \leq \varepsilon_n^{1+d/\beta} \right\}\right)\\
    & \geq \rho(k_n) \left(\frac{\varepsilon_n^{1+d/\beta}}{d_{k_n}}\right)^{\kappa d_{k_n}}.
\end{align*}
Combining assumptions \textbf{A2} and \textbf{A3}, there exist positive constants $A$ and $B$ such that 
$$
\rho(k_n) \geq \left(\frac{1}{d_{k_n}}\right)^{A d_{k_n}}\quad \text{and}\quad d_{k_n} \leq B \varepsilon_n^{-d/\beta}.
$$
It follows that for $n$ sufficiently large and taking $A > \kappa$,
\begin{align*}
    \rho(k_n) \left(\frac{\varepsilon_n^{2+d/\beta}}{d_{k_n}}\right)^{\kappa d_{k_n}}
    & \geq \left(\frac{1}{d_{k_n}}\right)^{A d_{k_n}} \left(\frac{\varepsilon_n^{1+d/\beta}}{d_{k_n}}\right)^{\kappa d_{k_n}}\\
    & \geq \left(\frac{\varepsilon_n^{1+2d/\beta}}{B}\right)^{AB \varepsilon_n^{-d/\beta}}\\
    & \geq \exp\left\{ - c_2 n \varepsilon_n^2 \right\}
\end{align*}
for some positive constant $c_2 >0$. This finishes the proof of Theorem \ref{thm:adaptative_convergence}.

\section{Auxiliary results}

\begin{lemma}\label{lem:approx}
Let $\mu$ be a finite measure on the compact metric space $(\ms,d)$. For each 
$n\geq 0$, $d_n\geq 0$, let $\{\phi_{i,n}\}_{i=0}^{d_n}$ be a set of densities 
(with respect to $\mu$) and let $\{R_{i,n}\}_{i=0}^{d_n}$ be a partition of 
$\ms$. Let $T_nf=\sum_{i=0}^{d_n} \left(\int_{R_{i,n}}f\, d\mu\right) 
\phi_{i,n}$, $f\in L^1(\ms)$. If the three following conditions hold:
\begin{enumerate}[(i)]
\item $\max_i \diam(R_{i,n}) \to 0$, as $n\to \infty$, where 
$\diam(R_{i,n})=\sup\{d(x,y): x,y \in R_{i,n}\}$,
\item for all $\delta > 0$, $\sum_{
\{i : d(x, R_{i,n}) \geq \delta\}} 
\mu(R_{i,n}) \phi_{i,n}(x) \rightarrow 
0$, uniformly in $x \in \ms$, where $d(x, R_{i,n}) := \inf\{d(x,y): y \in 
R_{i,n}\}$,
\item $\sum_{i=0}^{d_n}\mu(R_{i,n}) \phi_{i,n} = 1$, so that $T_nc=c$, for all 
$c\in \reals$,
\end{enumerate}
then we have $\|T_n f -f\|_\infty \to 0$ for every continuous density 
$f$.
\end{lemma}

\begin{proof}
Let $f$ be a 
(uniformly) continuous density on $\ms$ and let $\varepsilon > 0$. From 
\textit{(iii)} we have $|T_n f (x) - f(x)|\le \sum_{i=0}^{d_n} \int_{R_{i,n}} 
|f(y)-f(x)|\,\mu(dy) \phi_{i,n}(x)$. Take $\eps>0$, there exists $\delta > 0$, 
such that $|f(y) - f(x)| < \varepsilon /2$, for all  $y\in 
\text{B}_d(x,\delta)$. Using \textit{(i)}, let $N\geq 0$ be chosen so that 
$\max_i \diam(R_{i,n})< \delta/2$, for all $n\geq N$. Notice that for $n\geq 
N$, we have $\ms=\text{B}_d(x,\delta)\cup_{\{i: d(x, R_{i,n}) \geq 
\delta/2\}}R_{i,n}$; this follows from the fact that $d(x,y)\leq 
d(x,S)+\diam(S)$, for all $y\in S \subset \ms$. Therefore,
\begin{align*}
|T_n f (x) - f(x)| &\leq \sum_{i=0}^{d_n}\int_{R_{i,n}} 
|f(y)-f(x)|\,\mu(dy) \phi_{i,n}(x), \\
&\leq 
\frac{\eps}{2}\sum_{i=0}^{d_n}\int_{R_{i,n}\cap 
\text{B}_d(x,\delta)}\mu(dy) \phi_{i,n}(x)\\
&\qquad+2\|f\|_{\infty}\sum_{\{i: d(x, 
R_{i,n}) \geq \delta/2\}}\int_{R_{i,n}}\mu(dy) \phi_{i,n}(x),\\
&<\eps, \quad x \in \mathbb{M},
\end{align*}
follows from \textit{(iii)} and \textit{(ii)} provided $N$ is further chosen large enough.
\end{proof}

\begin{lemma}\label{lem:sum_order}
    If $a, b \in (0,\infty)$, then as $n \rightarrow \infty$ we have
    $$
        \log \sum_{j=1}^{\infty}  \left(\frac{j^b}{n^a}\right)^{-j}= \mathcal{O}\left(n^{a/b}\right).
    $$
\end{lemma}
\begin{proof}
    Let $k_n = n^{\gamma/b}$ for some $\gamma > a$ and write
    \begin{align*}
        \sum_{j=1}^{\infty} \left(\frac{j^b}{n^a}\right)^{-j} \leq \sum_{j > k_n}\left(\frac{j^b}{n^a}\right)^{-j}  +   k_n \max_{1 \leq j \leq k_n}\left(\frac{j^b}{n^a}\right)^{-j}.
    \end{align*}
    The second term on the right hand side is easily seen to be bounded by $k_n\exp\left( b n^{a/b}/e \right)$ and the first term is bounded by
    $
        \sum_{j=0}^{\infty} \left(\frac{k_n^b}{n^a} \right)^{-j} = \frac{1}{1-n^{a-\gamma}} \xrightarrow{n \rightarrow \infty} 1.
    $
    Taking the logarithm and neglecting low order terms then yields the result.
\end{proof}

\end{appendices}
\bibliographystyle{chicago}
\bibliography{biblio}
\end{document}